\documentclass{article}

\usepackage[letterpaper, portrait, margin=1in]{geometry}

\usepackage{url}
\usepackage{amsmath,amssymb,amsthm}
\usepackage{mathtools}
\usepackage{pgf}
\usepackage{tikz}
\usepackage{parskip}
\usepackage[colorlinks]{hyperref}

\usepackage{dsfont}

\newcommand{\ignore}[1]{}

\theoremstyle{plain}
\newtheorem{theorem}{Theorem}[section]
\newtheorem{lemma}[theorem]{Lemma}

\newtheorem{proposition}[theorem]{Proposition}
\newtheorem{claim}[theorem]{Claim}
\newtheorem{corollary}[theorem]{Corollary}
\newtheorem{question}[theorem]{Question}

\theoremstyle{definition}
\newtheorem{definition}[theorem]{Definition}
\newtheorem{remark}[theorem]{Remark}

\theoremstyle{theorem}
\newtheorem*{propositionref}{Proposition}
\newtheorem*{theoremref}{Theorem}
\newtheorem*{corollaryref}{Corollary}
\DeclareMathOperator{\poly}{poly}

\DeclareMathOperator{\rank}{rank}
\DeclareMathOperator{\depth}{depth}

\newcommand{\dist}{\mathsf{dist}}

\newcommand{\Exu}[2]{\underset{#1} \bE \left[ #2 \right] }

\renewcommand{\Pr}[1]{\bP \left[ #1 \right]}
\newcommand{\Pru}[2]{\underset{ #1 }\bP \left[ #2 \right]}

\newcommand{\define}{\vcentcolon=}

\renewcommand{\exp}[1]{\mathrm{exp}\left( #1 \right)}

\newcommand{\ceil}[1]{\ensuremath{\lceil #1 \rceil}}

\newcommand{\ind}[1]{\mathds{1} \left[ #1 \right] }

\newcommand{\zo}{\{0,1\}}

\newcommand{\cC}{\ensuremath{\mathcal{C}}}
\newcommand{\cD}{\ensuremath{\mathcal{D}}}
\newcommand{\cF}{\ensuremath{\mathcal{F}}}
\newcommand{\cG}{\ensuremath{\mathcal{G}}}

\newcommand{\cM}{\ensuremath{\mathcal{M}}}
\newcommand{\cP}{\ensuremath{\mathcal{P}}}

\newcommand{\cT}{\ensuremath{\mathcal{T}}}
\newcommand{\cU}{\ensuremath{\mathcal{U}}}

\newcommand{\cY}{\ensuremath{\mathcal{Y}}}
\newcommand{\cX}{\ensuremath{\mathcal{X}}}

\newcommand{\bE}{\ensuremath{\mathbb{E}}}
\newcommand{\bF}{\ensuremath{\mathbb{F}}}
\newcommand{\bN}{\ensuremath{\mathbb{N}}}
\newcommand{\bP}{\ensuremath{\mathbb{P}}}
\newcommand{\bR}{\ensuremath{\mathbb{R}}}

\newcommand{\ob}{\mathsf{univ}}
\newcommand{\weak}{\mathsf{weak}}
\newcommand{\cov}{\mathrm{cov}}
\newcommand{\ADJ}{\textsc{Adj}}
\newcommand{\embed}{\sqsubset}
\newcommand{\rembed}{\sqsubset}
\newcommand{\etal}{\emph{et al.}}

\title{Universal Communication, Universal Graphs, and Graph Labeling}

\author{Nathaniel Harms\thanks{University of Waterloo.
\texttt{nharms@uwaterloo.ca}}}

\begin{document}
\maketitle

\begin{abstract}
  We introduce a communication model called \emph{universal SMP}, in which Alice
  and Bob receive a function $f$ belonging to a family $\cF$, and inputs $x$ and
  $y$.  Alice and Bob use shared randomness to send a message to a third party
  who cannot see $f$, $x$, $y$, or the shared randomness, and must decide
  $f(x,y)$. Our main application of universal SMP is to relate communication
  complexity to graph labeling, where the goal is to give a short label to each
  vertex in a graph, so that adjacency or other functions of two vertices $x$
  and $y$ can be determined from the labels $\ell(x), \ell(y)$.  We give a
  universal SMP protocol using $O(k^2)$ bits of communication for deciding
  whether two vertices have distance at most $k$ in distributive lattices
  (generalizing the $k$-Hamming Distance problem in communication complexity),
  and explain how this implies a $O(k^2\log n)$ labeling scheme for deciding
  $\dist(x,y) \leq k$ on distributive lattices with size $n$; in contrast, we
  show that a universal SMP protocol for determining $\dist(x,y) \leq 2$ in
  modular lattices (a superset of distributive lattices) has super-constant
  $\Omega(n^{1/4})$ communication cost. On the other hand, we demonstrate that
  many graph families known to have efficient adjacency labeling schemes, such
  as trees, low-arboricity graphs, and planar graphs, admit constant-cost
  communication protocols for adjacency. Trees also have an $O(k)$ protocol for
  deciding $\dist(x,y) \leq k$ and planar graphs have an $O(1)$ protocol for
  $\dist(x,y) \leq 2$, which implies a new $O(\log n)$ labeling scheme for the
  same problem on planar graphs.
\end{abstract}

\section{Introduction}

In the simultaneous message passing (SMP) model of communication, introduced by
Yao \cite{Yao79}, Alice and Bob separately receive inputs $x$ and $y$ to a
function $f$. They send messages $a(x),b(y)$ to a third party, called the
referee, who knows $f$ and must output $f(x,y)$ (with high probability) using
the messages $a(x),b(y)$. But what if the referee \emph{doesn't} know $f$? Can
they still compute $f(x,y)$? Yes: Alice can include in her message a description
of $f$, and then the referee knows it; however, if $f$ is restricted, they can
sometimes do much better. Here is a simple example: the players receive vertices
$x,y \in \{1,\dotsc,n\}$ in a graph $G$ of maximum degree 2, and want to decide
if $(x,y)$ is an edge in $G$.  Sharing a source of randomness, Alice and Bob
randomly label each vertex of $G$ with a number up to 200; Alice sends the label
of both neighbors of $x$ and Bob sends the label of $y$.  The referee says
\emph{yes} if one of Alice's labels matches the label of $y$, \emph{no}
otherwise. They will be correct with probability at least $99/100$, and the
referee never needs to learn $G$. This is also an example where the referee can
decide many problems using only one strategy. In this work we will see that more
interesting families of graphs, such as trees, planar graphs, and distributive
lattices, also exhibit these phenomena, even when we wish to compute
\emph{distances} instead of just adjacency.

To study this, we introduce the \emph{universal SMP} model, which operates as
follows. Fix some family $\cF$ of functions. Alice and Bob receive a function $f
\in \cF$ and inputs $x,y$, and they use shared randomness to each send one
message to the referee. The referee knows the family $\cF$ and the size of the
inputs, but doesn't know $f,x,y$ or the shared randomness, and must compute
$f(x,y)$ with high probability. By choosing the family $\cF$ to be the singleton
family, one sees that this model includes standard SMP. As in the earlier
example, we will be studying communication problems on graphs, but this is not a
significant restriction: every Boolean-valued communication problem $f$ is
equivalent to determining adjacency in some graph (use $f$ as the adjacency
matrix), so we will treat $\cF$ as a family of graphs.

A surprising but intuitive application of universal SMP is that it connects two
apparently disjoint areas of study: communication complexity and graph labeling.
For a graph family $\cF$, the graph labeling problem (introduced by Kannan,
Naor, and Rudich \cite{KNR92}) asks how to assign the shortest possible labels
$\ell(v)$ to each vertex $v$ of a graph $G \in \cF$, so that the adjacency (or
some other function \cite{Pel05}) of vertices $x,y$ can be computed from
$\ell(x),\ell(y)$ by a decoder that knows $\cF$. We observe the following
principle (Theorem \ref{thm:universal smp to adjacency labeling}):

\emph{If there is a (randomized) universal SMP protocol for the graph family
$\cF$ with communication cost $c$, then there is a labeling scheme for graphs $G \in \cF$ with
labels of size $O(c \log n)$, where $n$ is the number of vertices.}

Common variants of graph labeling are \emph{distance labeling} \cite{GPPR04},
where the goal is to compute $\dist(x,y)$ from the labels, and
\emph{small-distance} labeling, where the goal is to compute $\dist(x,y)$ if it
is at most $k$ and output ``$>k$'' otherwise \cite{KM01, ABR05}. This is similar
to the well-studied $k$-Hamming Distance problem in communication complexity,
where the players must decide if their vertices $x,y$ have distance at most $k$
in the Boolean hypercube graph. A natural generalization of the Boolean
hypercube is the family of distributive lattices (which also include, for
example, the hypergrids). We demonstrate that techniques from communication
complexity can be used to obtain new graph labelings, by adapting the
$k$-Hamming Distance protocol of Huang \etal~\cite{HSZZ06} to the universal SMP
model, achieving an $O(k^2)$ protocol for computing $\dist(x,y) \leq k$ and the
corresponding $k$-distance labeling scheme with label size $O(k^2 \log n)$.  It
is interesting to note that, in contrast to the standard application of
communication complexity as a method for obtaining lower bounds, we are using it
to obtain upper bounds.

Generalizing in another direction, we ask: for which graphs other than the
Boolean hypercube can we obtain efficient communication protocols for
$k$-distance? For constant $k$, $k$-Hamming Distance can be computed with
communication cost $O(1)$; which other graphs admit a constant-cost protocol? To
approach this question, we observe that many (but not all) graph families known
to have efficient $O(\log n)$ adjacency labeling schemes also admit an $O(1)$
universal SMP protocol for adjacency. Commonly studied families in the adjacency
and distance labeling literature are trees \cite{KNR92, KM01, ABR05, AGHP16,
ADK17} and planar graphs \cite{KNR92,GPPR04,GL07,GU16,AKTZ19}. We study the
$k$-distance problem on these families and find that trees admit an $O(k)$
protocol, while planar graphs admit an $O(1)$ protocol for 2-distance; this
implies a new labeling scheme for planar graphs.

Further motivation for the universal SMP model comes from \emph{universal
graphs}. Introduced by Rado \cite{Rado64}, an induced-universal graph $U$ for a
set $\cF$ is one that contains each $G \in \cF$ as an induced subgraph. An
efficient adjacency labeling scheme for a set $\cF$ implies a small
induced-universal graph for that set \cite{KNR92}.  Deterministic universal SMP
protocols are equivalent to universal graphs (Theorem \ref{thm:complexity}), and
we introduce \emph{probabilistic universal graphs} as the analogous objects for
randomized univeral SMP protocols. We think probabilistic universal graphs are
worthy of study alongside universal graphs, especially since many non-trivial
families admit one of \emph{constant-size}.

The universal SMP model is also related to a recent line of work studying
communication between parties with imperfect knowledge of each other's
``context''. The most relevant incarnation of this idea is the recent work
\cite{GS17,GKKS18}, who study the 2-way communication model where Alice and Bob
receive functions $f$ and $g$ respectively, with inputs $x$ and $y$, and must
compute $f(x,y)$ under the guarantee that $f$ and $g$ are close in some metric.
In other words, one party does not have full knowledge of the function to be
computed. The universal SMP model provides a framework for studying a similar
problem in the SMP setting, where the players know the function but the referee
does not; the similarity is especially clear when we define the family $\cF$ to
be all graphs of distance $\delta$ to a reference graph $G$ in some metric (we
discuss this situation in more detail at the end of the paper). This could
model, for example, a situation where the clients of a service operate in a
shared environment but the server does not; or, a situation in which the clients
want to keep their shared environment secret from the server, and their inputs
secret from each other. This suggests a possible application to privacy and
security. A relevant example is private proximity testing (e.g.~\cite{NTL+11}),
where two clients should be notified by the server when they are at distance at
most $k$ from each other, without revealing to each other or the server their
exact locations.

The Discussion at the end of the paper highlights some interesting questions and
open problems.

\subsection{Results}

A universal SMP protocol \emph{decides $k$-distance} for a family $\cF$ if for
all graphs $G \in \cF$ and vertices $x,y$, the protocol will correctly decide if
$\dist(x,y) \leq k$, with high probability. A labeling scheme decides
$k$-distance if $\dist(x,y) \leq k$ can be decided from the labels of $x,y$.
Below, the variable $n$ always refers to the number of vertices in the input
graph.

\subparagraph*{Implicit graph representations.}
The main principle connecting communication and graph labeling is:
\begin{theorem}
\label{thm:universal smp to adjacency labeling}
Any graph family $\cF$ with universal SMP cost $m$ has an adjacency labeling
scheme with labels of size $O(m \log n)$. In particular, if the universal SMP
cost for $\cF$ is $O(1)$ then $\cF$ has an $O(\log n)$ adjacency labeling
scheme.
\end{theorem}
Adjacency labeling schemes of size $O(\log n)$ are of special interest because
$\log n$ is the minimum number of bits required to label each vertex uniquely,
and they correspond to \emph{implicit graph representations}, as defined by
Kannan, Naor, and Rudich \cite{KNR92} (we omit their requirement that the
encoding and decoding be computable in polynomial-time). Section
\ref{subsection:implicit graphs} elaborates further. To obtain implicit
representations, we can relax our requirements:

\begin{corollary}
\label{cor:2-way to implicit}
For any constant $c$, any graph family $\cF$ where each $G \in \cF$ has a
public-coin 2-way communication protocol computing adjacency with cost $c$ has
an implicit representation.
\end{corollary}

\subparagraph*{Distributive \& Modular Lattices.} Distributive and modular lattices are
generalizations of the Boolean hypercube and hypergrids (see Section
\ref{section:distributive lattices} for definitions). We define a
\emph{weakly-universal} SMP protocol as one where the referee shares the
randomness of Alice and Bob. For distributive lattices we get the following:

\begin{theorem}
\label{thm:distributive lattices}
The $k$-distance problem on the family of distributive lattices has: a
weakly-universal SMP protocol with cost $O(k\log k)$; a universal SMP protocol
with cost $O(k^2)$; and a size $O(k^2 \log n)$ labeling scheme.
\end{theorem}

Modular lattices are a superset of distributive lattices, but they do not admit
$k$-distance protocols with a cost independent of $n$; we show that any
universal SMP protocol (and any labeling scheme) deciding 2-distance must have
cost $\Omega(n^{1/4})$ (Theorem \ref{thm:lower bound for modular lattices}).  To
our knowledge, there are no known labeling schemes for distributive or modular
lattices. Our adjacency labeling scheme (i.e.~for $k=1$) requires $O(n \log n)$
space to store the whole lattice; this can be compared to Munro and Sinnamon
\cite{MS18}, who present a data structures of size $O(n \log n)$ for
distributive lattices that supports \emph{meet} and \emph{join} operations (and
therefore distance queries, due to our Lemma \ref{lemma:distance in modular
lattices}). However, these are not labelings, so the result is not directly
comparable.

\subparagraph*{Planar graphs and other efficiently-labelable families.} When they
introduced graph labeling, Kannan, Naor, and Rudich \cite{KNR92} studied trees,
low-arboricity graphs (whose edges can be partitioned into a small number of
trees), and planar graphs, and interval graphs (whose vertices are intervals in
$\bR$, with an edge if the intervals intersect), among others. These families
have $O(\log n)$ adjacency labeling schemes. Trees, low-arboricity graphs, and
planar graphs have constant-cost universal SMP protocols for adjacency. Trees
admit an efficient $k$-distance protocol: 
\begin{theorem}
\label{thm:k distance for trees}
  The family of trees has a universal SMP protocol deciding $k$-distance with
  cost $O(k)$ and a $O(k \log n)$ labeling scheme deciding $k$-distance.
\end{theorem}
Planar graphs admit an efficient 2-distance protocol, which implies a new
2-distance labeling scheme:
\begin{theorem}
\label{thm:k distance for planar graphs}
  The 2-distance problem on the family of planar graphs has a universal SMP
  protocol with cost $O(1)$ and a labeling scheme of size $O(\log n)$.
\end{theorem}
On the other hand, a universal SMP protocol deciding 2-distance on the family
of graphs with arboricity 2 has cost at least $\Omega(\sqrt n)$ (Proposition
\ref{prop:lower bound for arboricity 2 graphs}), and a universal SMP protocol
deciding adjacency in interval graphs has cost $\Theta(\log n)$ (Proposition
\ref{prop:lower bound for interval graphs}).

Gavoille \etal~\cite{GPPR04} showed that trees have an $O(\log^2 n)$
labeling allowing $\dist(x,y)$ to be computed exactly from labels of $x,y$, and
gave a matching lower bound; Kaplan and Milo \cite{KM01} and Alstrup \emph{et
al} \cite{ABR05} studied $k$-distance for trees, with the latter achieveing a
$\log n + O(k^2(\log\log n + \log k))$ labeling scheme. For planar graphs,
\cite{GPPR04} gives a lower bound of $\Omega(n^{1/3})$ for computing distances
exactly, and an upper bound of $O(\sqrt n \log n)$, which was later improved to
$O(\sqrt n)$ in \cite{GU16}.

\subparagraph*{Communication Complexity.} Our lower bounds are achieved by reduction
from the family of all graphs, which has complexity $\Theta(n)$, in contrast
to the upper bound of $\ceil{\log n}$ for the standard SMP cost of computing
adjacency in any graph (since Alice and Bob can send $\ceil{\log n}$ bits to
identify their vertices).
\begin{theorem}
\label{thm:lower bound for all graphs}
For the family $\cG$ of all graphs, the universal SMP cost of computing
adjacency in $\cG$ is $\Theta(n)$.
\end{theorem}
The basic relationships between universal SMP, standard SMP, and universal
graphs are as follows. Below, we use $D^\|(\ADJ(G))$ and $R^\|(\ADJ(G))$ for the
deterministic and randomized (standard) SMP cost of computing adjacency on $G$,
and $D^\ob(\cF), R^\ob(\cF)$ for the deterministic and randomized universal SMP
cost for computing adjacency in the family $\cF$.  We use the term
``$\embed$-universal graph'' as opposed to ``induced-universal'' to denote a
slightly different object that allows non-injective embeddings (see Section
\ref{section:universal smp} for definitions).
\begin{theorem}
\label{thm:complexity}
For a set $\cF$, the following relationships hold. Let $U$ range over the set of
all $\embed$-universal graphs:
\[
\max_{G \in \cF} D^\|(\ADJ(G)) \leq D^\ob(\cF_i) = \min_U D^\|(\ADJ(U))
  = \min_U \ceil{\log |U|} \,,
\]
with equality on the left iff $\exists H \in \cF$ such that $\forall G \in \cF$,
$G$ can be \emph{embedded} in $H$. For $\widetilde U$ ranging over the set of
all \emph{probabilistic universal graphs}:
\[
    \max_{G \in \cF} R^\|(\ADJ(G))
    \leq R^\ob(\cF)
    \leq \min_{\widetilde U} D^\|(\ADJ(\widetilde U)) 
    \leq O\left( R^\ob(\cF) \right) \,.
\]
Randomized and deterministic universal SMP satisfy
\[
  \Omega\left(\frac{D^\ob(\cF)}{\log n}\right)
  \leq R^\ob(\cF) \leq D^\ob(\cF) \,.
\]
\end{theorem}
The above results on graph labeling are proved through the relationship between
randomized and deterministic universal SMP. We obtain this relationship by
adapting Newman's Theorem \cite{New91}, a standard derandomization result in
communication complexity.  Finally, we note the interesting fact that universal
SMP characterizes the gap between standard SMP models where the referee does or
does not share the randomness with Alice and Bob:
\begin{proposition}[Informal]
\label{prop:weak to universal}
  Let $\cF$ be a family of graphs and let $\Pi$ be a weakly-universal SMP
  protocol for $\cF$, which defines a distribution over the referee's decision
  functions $F$, which we interpret as the adjacency matrices of graphs. Let
  $\cU_\Pi$ be the family on which this distribution is supported. Then, taking
  the minimum over all such protocols $\Pi$,
  \[
    R^\ob_{\epsilon}(\cF)
    = \min_{\Pi} D^\ob(\cU_\Pi) \,.
  \]
\end{proposition}

\subsection{Other Related Work}

\paragraph*{Graph labeling.} Randomized labeling schemes for trees have been
studied by Fraigniaud and Korman \cite{FK09}, who give a randomized adjacency
labeling scheme of $O(1)$ bits per label that has one-sided error (i.e.~it can
erroneously report that $x,y$ are adjacent when they are not), and they show
that achieving one-sided error in the opposite direction requires a randomized
labeling with $\Omega(\log n)$ bits.  They also give randomized schemes for
determining if $x$ is an ancestor of $y$, but they do not address distance
problems. Spinrad's book \cite{Spin03} has a chapter on implicit graphs and
Alstrup \etal~\cite{AKTZ19} for a recent survey on adjacency labeling schemes
and induced-universal graphs.  We know of no labeling schemes for lattices, but
Fraigniaud and Korman \cite{FK16} recently studied adjacency
labeling schemes for posets of low ``tree-dimension''.

\paragraph*{Distance-preserving} labeling studies an opposite problem to
$k$-distance labeling, where distances must be accurately reported when they are
\emph{above} some threshold $D$. Recent work includes Alstrup
\etal~\cite{ADKP16}.

To our knowledge, $k$-distance or even 2-distance has not been studied for
planar graphs, but there are many results on other types of planar graph
labelings  with restrictions at distance 2. An example is the \emph{frequency
assignment problem} or \emph{$L(p,q)$-labeling} problem, which asks how to
construct a labeling $\ell$ assigning integers $[k]$ to vertices of a planar
graph so that $\dist(x,y) \leq 1 \implies |\ell(x)-\ell(y)| \geq p$ and
$\dist(x,y) \leq 2 \implies |\ell(x)-\ell(y)| \geq q$, with various optimization
goals. See \cite{Cal11} for a survey.

\paragraph*{Uncertain communication.} There are several works studying
communication problems where the parties do not agree on the function to be
computed, starting with Goldreich, Juba, and Sudan \cite{GJS12} who studied
communication where parties have different ``goals''. Canonne
\etal~\cite{CGMS17} study communication in the shared randomness setting where
the randomness is shared imperfectly. Haramarty and Sudan \cite{HS16} study
compression (\'a la Shannon) in situations where the parties do not agree on a
common distribution. As mentioned earlier, Ghazi \etal~\cite{GKKS18} and Ghazi
and Sudan \cite{GS17} study 2-way communication where the parties do not agree
on the function to be computed.

\subsection{Notation}
$[k]$ means $\{1, \dotsc, k\}$.  The letter $n$ always denotes the number of
vertices in a graph. We use the notation $\ind{E} = 1$ iff the statement $E$
holds, and $\ind{E} = 0$ otherwise.  For a graph $G$, $V(G)$ is the set of
vertices and $E(G)$ is the set of edges. For vertices $x,y$, we write $G(x,y) =
\ind{x,y \text{ are adjacent in $G$}}$ for the entry in the adjacency matrix of
$G$. For an undirected, unweighted graph $G$ and vertices $u,v, \dist(u,v)$ is
the length of the shortest path from $u$ to $v$.

For any graph $G$ and integer $k$, we denote by $G^k$ the \emph{$k$-closure} of
$G$, where two vertices $u,v$ are adjacent iff $\dist(u,v) \leq k$ in $G$; it is
convenient to require that each vertex is adjacent to itself in $G^k$. For a set
of graphs $\cF$, $\cF^k = \{ G^k : G \in \cF \}$.

$D^\|(f)$ is the deterministic SMP cost of the function $f$ and $R^\|(f)$ is the
randomized SMP cost of the function $f$, in the model where Alice and Bob share
randomness but the deterministic referee does not.

\section{Universal Communication and Universal Graphs}
\label{section:universal smp}
In this paper we focus on deciding adjacency. Every Boolean communication
problem $f : \cX \times \cY \to \zo$ on finite domains $\cX,\cY$ is equivalent
to the adjacency problem on the graph $G$ with vertex set $\cX \cup \cY$ and
$G(u,v) = f(u,v)$. We may either allow self-loops in $G$ if $\cX = \cY$ or take
$G$ to be bipartite. We will generally permit graphs to have self-loops.

\begin{definition}
A \emph{family of graphs} $\cF = (\cF_i)$ is a sequence of sets $\cF_i$ indexed
by integers $i$, along with a strictly increasing size function $n(i)$, so that
$\cF_i$ is a set of graphs with vertex set $[n(i)]$. If $\cF_i$ has size
$n(i)=i$ then we write $\cF_n$.  
\end{definition}

\begin{definition}[Universal SMP and Variations]
  Let $\cF$ be a family of graphs with size function $n$ and let $\Phi$ be an
  operation taking size $n(i)$ graphs to size $n(i)$ graphs. Let $c : \bN \to
  \bN$ and let $\epsilon > 0$ be a constant. An $\epsilon$-error, cost $c$
  sequence of \emph{universal SMP} communication protocols for $\cF$ is as
  follows. For any $i \in \bN$, a protocol $\Pi_i$ for $\cF_i$ is a triple
  $(a_i, b_i, F_i)$ where:
  \begin{itemize}
    \item Alice and Bob receive $(G,x), (G,y)$ respectively, where $G \in \cF_i$
      and $x,y \in V(G) = [n(i)]$;
    \item Alice and Bob share a random string $r$ and compute messages
      $a_i(r,G,x), b_i(r,G,y) \in \zo^{c(i)}$, respectively;
    \item For each $i$, the (deterministic) referee has a function $F_i :
      \zo^{c(i)} \times \zo^{c(i)} \to \zo$, called the \emph{decision
      function}. $F_i(a_i(r,G,x), b_i(r,G,y))$ must satisfy:
      \begin{enumerate}
        \item If $x,y$ are adjacent in $\Phi(G)$ then
          $\Pru{r}{F_i(a_i(r,G,x),b_i(r,G,y)) = 1} > 1-\epsilon$; and
        \item If $x,y$ are not adjacent in $\Phi(G)$ then
          $\Pru{r}{F_i(a_i(r,G,x),b_i(r,G,y))} < \epsilon$.
      \end{enumerate}
  \end{itemize}
  A universal SMP protocol is \emph{symmetric} when the functions $a_i,b_i$
  computed by Alice and Bob are identical and the function $F_i$ satisfies
  $F_i(a,b) = F_i(b,a)$ for all messages $a,b \in \zo^c$.  We write
  $R^\ob_\epsilon(\Phi(\cF))$ for the communication complexity in the universal
  SMP model of computing adjacency in graphs $\Phi(\cF) = \{ \Phi(G) : G \in \cF
  \}$, where $\epsilon$ is the allowed probability of error.  We write
  $R^\ob(\Phi(\cF))$ for $R^\ob_{1/3}(\Phi(\cF))$.  If no operation $\Phi$ is
  specified, it is assumed to be the identity.

  It is also convenient to define a \emph{weakly-universal SMP} protocol as a
  universal SMP protocol where the referee can see the shared randomness, so the
  choice function is of the form $F_i(r, a(r,G,x), b(r,G,y))$ for random seed
  $r$, graph $G \in \cF$, and $x,y \in V(G)$. We denote the $\epsilon$-error
  complexity in this model with $R^\weak_\epsilon(\Phi(\cF))$.

  Finally, we write $D^\ob(\Phi(\cF))$ for the \emph{deterministic} universal
  SMP complexity.
\end{definition}
\begin{remark}
  We include the operator $\Phi$ in the definition to emphasize that the players
  are given the original graph $G$, not the graph $\Phi(G)$; for example, the
  players are not given $G^k$ (from which it may be difficult to compute $G$),
  but are instead given $G$.
\end{remark}

\subsection{Deterministic Universal Communication and Universal Graphs}
\label{subsection:deterministic universal smp}
We will show that a deterministic universal SMP protocol is equivalent to an
\emph{embedding} into a $\embed$-universal graph, which we we define using the
following notion of embedding (following the terminology of Rado \cite{Rado64}):
\begin{definition}
  For graphs $G,H$, a mapping $\phi : V(G) \to V(H)$ is an \emph{embedding} iff
  $\forall u,v \in V(G)$, $G(u,v) = H(\phi(u),\phi(v))$. If such a mapping
  exists we write $G \embed H$.

  For a set of graphs $\cF_i$, a graph $U$ is \emph{$\embed$-universal} if
  $\forall G \in \cF_i, G \embed U$; i.e.~$\forall G \in \cF_i$ there exists an
  embedding $\phi_G : V(G) \to V(U)$.
  For a family of graphs $\cF = (\cF_i)$, a sequence $U = (U_i)$ is a
  \emph{$\embed$-universal graph sequence} if for each $i$, $U_i$ is
  $\embed$-universal for $\cF_i$.

  Define an equivalence relation on $V(G)$ by $u \equiv v$ iff $\forall w \in
  V(G), G(u,w) = G(v,w)$, i.e.~$u,v$ have identical rows in the adjacency
  matrix. For a graph $G$, define the $\equiv$-reduction $G^\equiv$ as a graph
  on the equivalence classes $\cC$ of $V(G)$ with $U,W \in \cC$ adjacent iff
  $\exists u \in U, w \in W$ such that $u,w$ are adjacent.
\end{definition}
An embedding is not the same as a homomorphism since we must map non-edges to
non-edges, and $G \embed H$ is not the same as $G$ being an induced subgraph of
$H$ since the mapping is not necessarily injective. Therefore a universal graph
by our definition is not the same as an induced-universal graph, where $G$ must
exist as an induced subgraph.  We could for example map the path $a$ --- $b$ ---
$c \mapsto a'$ --- $b'$ --- $a'$. This difference between definitions is
captured by the $\equiv$ relation between vertices. It is necessary to allow
self-loops, otherwise the $\embed$ relation is not transitive. The important
properties of $\embed, \equiv$, and $\equiv$-reductions are stated in the next
proposition; the proofs are routine and for completeness are included in the
appendix. The relation $\simeq$ is the isomorphism relation on graphs.
\begin{proposition}
\label{prop:embedding properties}
  The following properties are satisfied by the $\embed$ relation, the $\equiv$
  relation, and $\equiv$-reductions:
  \begin{enumerate}
    \item $\embed$ is transitive.
    \item For any graph $G$ and $u,v \in V(G)$, $u \equiv v$ iff there exists
      $H$ and an embedding $\phi : G \to H$ such that $\phi(u) = \phi(v)$.
    \item For any graph $G, (G^\equiv)^\equiv \simeq G^\equiv$.
    \item For any graph $G, G \embed G^\equiv$ and $G^\equiv \embed G$.
    \item For any graphs $G,H$, $G \embed H$ iff $G^\equiv \embed H^\equiv$.
    \item For any graphs $G,H$, $G^\equiv \embed H^\equiv$ iff $G^\equiv$ is an
      induced subgraph of $H^\equiv$.
  \end{enumerate}
\end{proposition}
These properties allows us to prove relationships between the standard SMP
model, deterministic universal SMP, and $\embed$-universal graphs. First we show
that deterministic universal SMP protocols can always be made
symmetric\footnote{Note that this does not imply that every deterministic SMP
protocol is symmetric, since in this paper we are only concerned with adjacency
on an undirected graph, for which the communication matrix is symmetric. This
proposition shows that for symmetric communication matrices, the deterministic
SMP protocol is symmetric.}.
\begin{proposition}
  \label{prop:deterministic protocols are symmetric} If $\Pi$ is a deterministic
  universal SMP protocol for the set $\cF$, then there exists a deterministic
  universal SMP protocol $\Pi'$ that is symmetric and has the same cost as
  $\Pi$.
\end{proposition}
\begin{proof}
  Let $G \in \cF$ and let $a,b : V(G) \to \zo^m$ be the encoding functions for
  $G$ and $F$ the decision function for graphs of size $|G|$. The restriction of
  $b$ to the domain $V(G^\equiv) \to \zo^m$ is injective so it has an inverse
  $b^{-1} : \mathrm{image}(b) \to V(G^\equiv)$ that satisfies $b^{-1}b(x) \equiv
  x$; the same holds for $a, a^{-1}$.  Define the encoding function $b' : V(G)
  \to \zo^m$ as $b' = ab^{-1}b$ and define the decision function $F'(p,q) = F(p,
  ba^{-1} (q))$. Then for any $x,y \in V(G), F'(a(x), b'(y)) = F(a(x),
  ba^{-1}ab^{-1}b(y)) = F(a(x),b(y)) = G(x,y)$ so this is a valid protocol.
  Since $\mathrm{image}(b') \subseteq \mathrm{image}(a)$ we can write $b'(x) =
  aa^{-1}b'(x) = aa^{-1}ab^{-1}b(x) = a(x)$ for every $x$ so $b'=a$, thus
  $F'(a(x),a(y)) = G(x,y) = G(y,x) = F'(a(y),a(x))$ so the protocol is symmetric.
\end{proof}
The standard deterministic SMP complexity measure can be expressed in terms of
$\equiv$-reductions:
\begin{proposition}
  For all graphs $G$, $D^\|(\ADJ(G)) = \ceil{ \log |G^\equiv|}$.
\end{proposition}
\begin{proof}
  It is well-known that for any function $f : \cX \times \cY \to \zo$, $D^\|(f)
  = \ceil{ \log \min(r,c) }$ where $r$ is the number of distinct columns in the
  communication matrix of $f$, and $c$ is the number of distinct rows
  \cite{Yao79}. The communication matrix of the function $\ADJ(G)$ is the
  adjacency matrix of $G$, which is symmetric, and two rows (or columns) indexed
  by $u,v$ are distinct iff $u \not \equiv v$; so the number of distinct rows is
  the size of $G^\equiv$.
\end{proof}
The analogous fact for universal SMP is that the deterministic universal SMP
cost is determined by the size of the smallest universal graph.
\begin{proposition}
  For any graph family $\cF = (\cF_i)$,
  \[
    D^\ob(\cF_i)
    = \min_U \{ \ceil{ \log |U^\equiv| } : \forall G \in \cF_i, G \embed
    U^\equiv \} \,.
  \]
\end{proposition}
\begin{proof}
Let $U$ be any graph such that $G \embed U^\equiv$ for all $G \in \cF_i$ and for
each $G \in \cF_i$ let $g$ be the embedding $G \to U^\equiv$.
Consider the protocol where on inputs $(G,x),(G,y)$, Alice and Bob send
$g(x),g(y)$ using $\ceil{\log |U^\equiv|}$ bits and the referee outputs
$U^\equiv(g(x),g(y))$. This is correct by definition so $D^\ob(\cF_i) \leq
\ceil{\log|U^\equiv|}$.

Now suppose there is a protocol $\Pi$ for $\cF_i$ with cost $c$ and decision
function $F_i$, and let $G \in \cF_i$.  By Proposition \ref{prop:deterministic
protocols are symmetric} we may assume that on inputs $(G,x),(G,y)$ Alice and
Bob share the encoding function $g : V(G) \to \zo^c$. Let $U$ be the graph with
vertices $\zo^c$ and $U(u,v) = F(u,v)$. Then $U(g(x),g(y)) = F(g(x),g(y)) =
G(x,y)$ so $G \embed U \embed U^\equiv$ (by transitivity). Now $|U^\equiv| \leq
2^c$ so $c \geq \log |U^\equiv|$.
\end{proof}
It is easy to see that $D^\|$ can be used as a lower bound on $D^\ob$ but such
lower bounds are tight only when the family $\cF$ is essentially a ``trivial''
family of equivalent graphs.
\begin{lemma}
  \label{lemma:deterministic hierarchy}
  For any family $\cF = (\cF_i)$, let $U = (U_i)$ be the smallest
  $\embed$-universal graph sequence for $\cF$. Then
  \[
  \max_{G \in \cF_i} D^\|(\ADJ(G)) \leq D^\ob(\cF_i)
  = D^\|(\ADJ(U_i)) \,,
  \]
  with equality holding on the left iff $\exists H \in \cF_i$ such that $\forall
  G \in \cF_i, G^\equiv \embed H^\equiv$.
\end{lemma}
\begin{proof}
  The equality on the right holds by the two prior propositions.  The lower
  bound follows from the fact that any protocol $\Pi_i$ for $\cF_i$ in the
  universal model can be used as a protocol in the SMP model. Now we must show
  the equality condition.  Let $U \in \cF_i$ be a graph maximizing $|U^\equiv|$
  over all graphs in $\cF_i$, and suppose $D^\ob(\cF_i) = \max_{G \in \cF_i}
  D^\|(\ADJ(G)) = \max_{G \in \cF_i} \ceil{ \log |G^\equiv| } =
  \ceil{\log|U^\equiv|}$, so $\ceil{\log |U^\equiv|} = \min\{ \ceil{\log
  |H^\equiv|} : \forall G \in \cF_i, G \embed H^\equiv \}$. Then there exists
  $H$ such that $U^\equiv \embed H^\equiv$ and $|U^\equiv| = |H^\equiv|$.
  Since $U^\equiv$ is an induced subgraph of $H^\equiv$ and $|U^\equiv| =
  |H^\equiv|$ we must have $U^\equiv \simeq H^\equiv$ so $\forall G \in \cF_i,
  G^\equiv \embed U^\equiv$.
\end{proof}

\subsection{Randomized Universal Communication}
Just as deterministic universal communication is equivalent to embedding a
family into a universal graph, we will define probabilistic universal graphs and
show that they are tightly related to universal communication with shared
randomness.
\begin{definition}
  For graphs $G,H$, a random mapping $\phi : V(G) \to V(H)$ (i.e.~a distribution
  over such mappings) is an \emph{$\epsilon$-error embedding} iff $\forall u,v
  \in V(G)$,
  \[
    \Pru{\phi}{ G(u,v) = H(\phi(u),\phi(v)) } > 1-\epsilon \,.
  \]
  We will write $G \rembed_\epsilon H$ if there exists an $\epsilon$-error
  embedding $G \to H$. A graph $U$ is \emph{$\epsilon$-error universal} for a
  set of graphs $S$ if $\forall G \in S, G \rembed_\epsilon U$.  $U = (U_i)$ is
  an $\epsilon$-error universal graph sequence for the family $\cF = (\cF_i)$ if
  for each $i$, $U_i$ is $\epsilon$-error universal for $\cF_i$.
\end{definition}
In the randomized setting we obtain equivalence (up to a constant factor)
between universal SMP protocols and probabilistic universal graphs.
\begin{lemma}
\label{lemma:randomized symmetrization}
  For any graph family $\cF = (\cF_i)$ and any $\epsilon > 0$, if there exists a
  $\epsilon$-error universal SMP protocols for $\cF$ with cost $c(i)$, then
  there exists a $2\epsilon$-error symmetric universal SMP protocols for $\cF$
  with cost at most $2c(i)$.
\end{lemma}
\begin{proof}
  On input $G \in \cF_i, x,y \in V(G)$, and random string $r$, Alice and Bob
  send the concatentations $g_r(x) \define a_i(r,G,x) b_i(r,G,x)$ and $g_r(y)
  \define a_i(r,G,y) b_i(r,G,y)$. Then the referee computes
  \[
  F'_i( g_r(x), g_r(y) ) = \max \left\{
    F_i( a_i(r,G,x), b_i(r,G,y) ), F_i( a_i(r,G,y), b_i(r,G,x) )
    \right\} \,.
  \]
  It is clear that $F'_i$ is symmetric.  If $x,y$ are adjacent then
  \begin{align*}
    \Pru{r}{ F'_i(g_r(x),g_r(y)) = 0}
    \leq \Pru{r}{ F_i(a_i(r,G,x),b_i(r,G,y)) = 0 } < \epsilon \,,
  \end{align*}
  and if $x,y$ are not adjacent then, by the union bound,
  \begin{align*}
    &\Pru{r}{ F'_i(g_r(x),g_r(y)) = 1} \\
    &\qquad\leq \Pru{r}{ F_i(a_i(r,G,x),b_i(r,G,y)) = 1 } 
    + \Pru{r}{ F_i(a_i(r,G,y),b_i(r,G,x)) = 1 } < 2\epsilon \,. \qquad \qedhere
  \end{align*}
\end{proof}
Applying this symmetrization, we get a relationship between universal SMP
protocols and probabilistic universal graphs.
\begin{lemma}
\label{lemma:universal smp to universal graph}
Let $\cF = (\cF_i)$ be a graph family and $\epsilon > 0$. Then
\begin{enumerate}
  \item There is an $\epsilon$-error universal graph sequence of
    size at most $2^{2 R^\ob_{\epsilon/2}(\cF)}$; and
  \item If there is an $\epsilon$-error universal graph sequence
    of size $c(i)$ then $R^\ob_\epsilon(\cF) \leq \ceil{\log c}$.
\end{enumerate}
\end{lemma}
\begin{proof}
  If $\Pi_i$ is an $\epsilon$-error symmetric universal protocol for $\cF_i$
  then there exists a function $F_i$ such that for every $G \in \cF_i$ there is
  a random $g$ such that $\Pru{g}{F_i(g(x),g(y)) \neq G(x,y)} < \epsilon$. Using
  $F_i$ as an adjacency matrix, we get a graph $U_i$ of size at most $2^c$, where
  $c$ is the cost of $\Pi_i$, such that for all $G \in
  \cF_i, G \rembed_\epsilon U_i$. Then $U = (U_i)$ is an $\epsilon$-error
  probabilistic universal graph sequence. By Lemma \ref{lemma:randomized
  symmetrization} we obtain an $\epsilon$-error symmetric protocol with cost $2
  R^\ob_{\epsilon/2}(\cF)$, so we have proved the first conclusion. The second
  conclusion follows by definition.
\end{proof}
The basic relationships to standard SMP models follow essentially by definition
and from the above lemma.
\begin{lemma}
  \label{lemma:randomized hierarchy}
  Let $\cF$ be any graph family and let $\epsilon > 0$. Let $U =
  (U_i)$ be an $\embed$-universal graph sequence for $\cF$, and
  $\widetilde U = (\widetilde U_i)$ an $\epsilon$-error universal graph
  sequence. Then
  \[
    \max_{G \in \cF_i} R^\|_\epsilon(\ADJ(G))
    \leq R^\ob_{\epsilon}(\cF_i)
    \leq D^\|(\ADJ(\widetilde U_i))
    \leq 2R^\ob_{\epsilon/2}(\cF_i)
    \;\text{ and }\;
    R^\ob_{\epsilon}(\cF_i)
    \leq R^\|_\epsilon(\ADJ(U_i)) \,.
  \]
\end{lemma}
\begin{proof}
  The inequalities on the left follow the definitions and from the above lemma.
  On the right, we can obtain a universal SMP protocol by choosing for each $G
  \in \cF_i$ a (deterministic) embedding $g : G \to U_i$ and then using the
  randomized SMP protocol for $\ADJ(U_i)$.
\end{proof}
Universal graphs describe an interesting relationship between weakly-universal
and universal SMP protocols (and therefore between standard SMP protocols where
the referee does and does not share the randomness); namely, the optimal
universal protocol is obtained by finding the smallest universal graph for the
family of protocol graphs (decision functions) defined by a weakly-universal
protocol.

\begin{propositionref}[\ref{prop:weak to universal}]
Let $\cF$ be a family of graphs, let $\epsilon > 0$, and let $W_\epsilon$ be the set of all
$\epsilon$-error weakly-universal SMP protocols for $\cF$. For each $\Pi \in
W_\epsilon$ let $\cU_\Pi = (\cU_{\Pi,i})$ be the family of graphs $\cU_{\Pi,i}
= \{ F_i(r, \cdot, \cdot) : r \text{ is a random seed for } \Pi \}$ where
$F_i$ is the decision function of $\Pi$. Then
\[
  R^\ob_{\epsilon}(\cF)
  = \min_{\Pi \in W_\epsilon} D^\ob(\cU_\Pi) \,.
\]
\end{propositionref}
\begin{proof}
  Let $\Pi \in W_\epsilon$; we will construct a universal SMP protocol as
  follows. On input $(G,x),(G,y)$, Alice and Bob use shared randomness $r$ to
  simulate $\Pi$ and obtain vertices $a(r,G,x),b(r,G,y)$ in some graph $U_r \in
  \cU_\Pi$ with $\bP_{r}[U_r(a(r,G,x),b(r,G,y)) \neq G(x,y)] < \epsilon$. They
  now simulate the deterministic universal SMP protocol, i.e.~an embedding $\phi :
  V(U_r) \to U'$ for some graph $U'$ that is $\embed$-universal for $\{U_r\}$,
  and send $\phi(a(r,G,x)), \phi(b(r,G,y))$ to the referee who computes
  $U'(\phi(a(r,G,x)),\phi(b(r,G,x))) = U_r(a(r,G,x),b(r,G,y))$.

  Now let $\Pi$ be an $\epsilon$-error universal SMP
  protocol. Then $\Pi \in W_\epsilon$ and for each $i$, $\cU_{\Pi,i} = \{ U_i
  \}$, where $U_i$ is the graph of the decision function. $D^\ob(\cU_\Pi) \leq
  \ceil{\log |U_i|}$, which is the cost of $\Pi$, so $\min_{\Pi \in W_\epsilon}
  D^\ob(\cU_\Pi) \leq R^\ob_\epsilon(\cF)$.
\end{proof}
Newman's Theorem for public-coin randomized (2-way) protocols is a classic
result that gives a bound on the number of uniform random bits required to
compute a function $f : \cX \times \cY \to \zo$ in terms of the size of the
input domain \cite{New91}.  In the universal model, the input size can be very
large since the graph (function) itself is part of the input. However, the
shared part of the input does not contribute to the number of random bits
required in the universal SMP model.
\begin{lemma}[Newman's Theorem for universal SMP]
\label{lemma:newman's theorem for universal smp}
Let $\epsilon, \delta > 0$ and suppose there is an $\epsilon$-error universal
SMP protocol $\Pi$ for the family $\cF = (\cF_i)$. Then there is an
$(\epsilon+\delta)$-error universal SMP protocol for the family $\cF$ that uses
at most $\log \log \left(n(i)^{O(\epsilon/\delta^2)}\right)$ bits of randomness
and has the same communication cost.
\end{lemma}
\begin{proof}
  Fix $i$, let $F$ be the deterministic decision function for $\cF_i$, and let
  $a(r, \cdot, \cdot),b(r, \cdot, \cdot)$ be Alice and Bob's encoding functions
  for the random seed $r$. For $G \in \cF_i$ and $x,y \in V(G)$ we will say a
  seed $r$ is \emph{bad} for $G,x,y$ if $F(a(r,G,x),b(r,G,y)) \neq G(x,y)$, and
  we will call this event $\mathsf{bad}(G,x,y,r)$.

  Let $r_1, \dotsc, r_m$ be independent random seeds, and let $i \sim [m]$ be
  uniformly random, where $m > \frac{3\epsilon}{\delta^2}\ln(n^2)$. Then for
  every $G$, the expected number of vertex pairs $x,y$ for which the strings
  $r_1, \dotsc, r_m$ fail is
  \begin{align*}
    &\Exu{r_1, \dotsc, r_m}{
      \sum_{x,y} \ind{\Pru{i \sim [m]}{\mathsf{bad}(G,x,y,r_i)} > \epsilon + \delta }} \\
      &\qquad\leq n^2 \max_{x,y} \Exu{r_1, \dotsc, r_m}{ \ind{\Pru{i}{\mathsf{bad}(G,x,y,r_i)} >
      \epsilon + \delta } } \\
      &\qquad= n^2 \max_{x,y} \Pru{r_1, \dotsc, r_m}{ \Pru{i}{\mathsf{bad}(G,x,y,r_i)} > \epsilon
      + \delta} \\
      &\qquad= n^2 \max_{x,y} \Pru{r_1, \dotsc, r_m}{ \sum_{i=1}^m \ind{\mathsf{bad}(G,x,y,r_i)}
      > m(\epsilon+\delta)} \,.
  \end{align*}
  The sum has mean
  $\mu = \sum_{i=1}^m \Exu{r_i}{\ind{\mathsf{bad}(G,x,y,r_i)}} < m \epsilon$, so 
  by the Chernoff bound, the probability is at most
  \begin{align*}
    &n^2 \Pru{r_1, \dotsc, r_m}{\sum_{i=1}^m \ind{\mathsf{bad}(G,x,y,r_i)} >
    (1+m\delta/\mu)\mu} \\
    &\qquad\leq n^2 \exp{-\frac{m^2\delta^2}{3\mu}}
    \leq n^2 \exp{-\frac{m\delta^2}{3\epsilon}} 
    < 1 \,.
  \end{align*}
  Since the expected number of pairs $x,y$ where choosing $i \sim [m]$ fails
  with probability more than $\epsilon + \delta$ is less than 1, there must be
  some values of $r_1, \dotsc, r_m$ with no bad pairs for $G$. So for every $G
  \in \cF_i$ we may choose $r_1, \dotsc, r_m$ so that choosing $i$ uniformly at
  random is the only random step; since $m = \frac{6\epsilon}{\delta^2}\ln n =
  \log n^{O(\epsilon/\delta^2)}$ this requires at most $\log m = \log \log
  \left(n^{O(\epsilon/\delta^2)}\right)$ random bits.
\end{proof}
With this result, we can conclude the proof of Theorem \ref{thm:complexity} in
the next lemma.
\begin{lemma}
  \label{lemma:randomized to deterministic}
  For any family $\cF = (\cF_i)$ with size function $n(i)$,
  \[
    \Omega\left(\frac{D^\ob(\cF_i)}{\log n(i)}\right)
    \leq R^\ob(\cF_i)
    \leq D^\ob(\cF_i) \,.
  \]
\end{lemma}
\begin{proof}
  The upper bound is clear, so we prove lower bound. Let $\Pi = (\Pi_i)$ be a
  sequence of randomized universal SMP protocols for $\cF$.  By Newman's
  theorem, we may assume that $\Pi_i$ uses at most $\log \log n(i)^c$ random
  bits for some constant $c$ and has error probability $3/8$. Let $F_i$ be the
  decision function of $\Pi_i$, let $m(i)$ be the cost of $\Pi_i$, and let $k =
  \ceil{c \log n(i) }$.  To obtain a deterministic protocol, we can define the
  decision function $F'_i$ on messages of $k \cdot m(i)$ bits as $F'_i(a_1,b_1,
  a_2,b_2, \dotsc, a_k,b_k) = \mathrm{majority}( F_i(a_j,b_j))_j$. Alice and Bob
  iterate over all $k = 2^{\log \log n(i)^c}$ random strings $r$ and send $a(r,
  G, x), b(r,G,y)$ for each. Since the probability of error is at most $3/8$
  when $r$ is uniform, at least $5k/8 > k/2$ of the functions $F_i(a_j,b_j)$
  will give the correct answer. This proves that $D^\ob(\cF_i) = O( R^\ob(\cF_i)
  \log n(i) )$.
\end{proof}

In this paper we show lower bounds for a family $\cF$ by giving embeddings of an
arbitrary graph $G$ into $\cF$, so we need to know the complexity of the family
$\cG = (\cG_n)$ of all graphs with $n$ vertices. For our purposes, it is
convenient to require that each graph $G \in \cG_n$ has $G(u,u) = 1$ for all $u$
(i.e.~all self-loops are present). However, since equality can be checked with
cost $O(1)$, the presence or absence of self-loops does not affect the
complexity.

\begin{theoremref}[\ref{thm:lower bound for all graphs}]
$R^\ob(\cG) = \Theta(n)$.
\end{theoremref}
\begin{proof}
  For the upper bound, consider the (deterministic) protocol where on input
  $G,x,y$, Alice and Bob send $x$ and $y$ and the respective rows of the
  adjacency matrix of $G$. This has cost $n + \ceil{\log n} = O(n)$ and the
  referee can determine $G(x,y)$ by finding $y$ in the row sent by Alice.

  Let $\Pi$ be any protocol for $\cG_n$ with cost $c$.  By Lemma
  \ref{lemma:randomized symmetrization}, we may assume that $\Pi$ is symmetric.
  Let $F$ be the decision function for graphs on $n$ vertices
  and let $G \in \cG_n$ with vertex set $[n]$. $\Pi$ defines a distribution
  over functions $g : [n] \to \zo^c$ so that for all $x,y, \Pru{g}{F(g(x),g(y))
  \neq G(x,y)} < \epsilon$. Therefore, for $x,y$ drawn uniformly from $[n]$,
  $\Exu{f,x,y}{ \ind{F(f(x),f(y)) \neq G(x,y)} } < \epsilon$.  Therefore, for
  every graph $G \in \cG_n$ there is a function $f_G$ such that for $x,y \sim
  [n]$ uniformly at random, $\Pru{x,y}{F(f_G(x),f_G(y)) \neq G(x,y)} <
  \epsilon$. Write $N = {n \choose 2}$. There are at most $2^{cn}$ functions
  $[n] \to \zo^c$ and there are $2^N$ simple graphs on $[n]$ so there is some
  function $f : [n] \to \zo^c$ where the number of graphs $G$ such that $f_G =
  f$ is at least $\frac{2^N}{2^{cn}} = 2^{N - cn}$. Let $G,G'$ be any two such
  graphs. Then
  \begin{align*}
    &\Pru{x,y \sim [n]}{ G(x,y) \neq G'(x,y)} \\
    &\qquad\leq \Pru{x,y \sim [n]}{ G(x,y) \neq F(f(x),f(y)) \text{ or }
                             G'(x,y) \neq F(f(x),f(y)) }
    < 2\epsilon \,.
  \end{align*}
  So $G,G'$ differ on at most $2\epsilon N$ pairs. However, the largest number
  of graphs that differ from any graph $G$ on at most $2\epsilon N$ pairs of
  vertices is at most
  \[
    \sum_{k=0}^{2\epsilon N} {N \choose k}
    \leq 2\epsilon N {N \choose 2 \epsilon N}
    \leq \epsilon N \left(\frac{e N}{2 \epsilon N}\right)^{2 \epsilon N}
    = 2^{2 \epsilon N \log(e/2\epsilon) + \log(2 \epsilon N)} \,.
  \]
  Therefore we must have
  \[
    N - cn 
    \leq 2 \epsilon N \log(e/2\epsilon) + \log(2 \epsilon N)
  \]
  so $c = \Omega(n)$.
\end{proof}

Recall the example in the first paragraph of the introduction, for which we
observed that a single decision function would work for many problems. We now
make a note about this phenomenon.  A communication protocol for a graph family
$\cF = (\cF_i)$ is really a sequence of protocols, one for each set $\cF_i$ of
graphs with $n(i)$ vertices. Our next proposition addresses the uniformity of
the sequence of protocols, that is, the question of how the protocols are
related to one another as the size of the input grows. In general, we ask the
question: If the family $\cF$ has some relationship between $\cF_i$ and
$\cF_{i+1}$, what does this imply about the relationship between the protocols
for $i$ and $i+1$?  The families of graphs we study in this paper have
constant-cost protocols and they are also \emph{upwards families}, which we
define next. These families have enough structure so that there exists a single,
one-size-fits-all probabilistic universal graph, into which all graphs can be
embedded regardless of their size; in other words, the referee can be ignorant
not only of the graph $G$ and vertices $x,y$, but also of the \emph{size} of the
graph, without increasing the cost of the protocol.\footnote{Any family $\cF$
with a constant-cost protocol can be turned into a protocol ignorant of the size
by requiring that Alice and Bob tell the referee which of the $2^{c^2}$ possible
decision functions to use, where $c = 2^{R^\ob(\cF)}$.}
\begin{definition}
  We call a graph family $\cF = (\cF_i)$ an \emph{upwards family} if for every
  $i$ and every $G \in \cF_i$ there exists $G' \in \cF_{i+1}$ such that $G$ is
  an induced subgraph of $G'$.
\end{definition}
Many graph families are upwards families, for example: bounded-degree graphs,
bounded-arboricity graphs, planar graphs, and transitive reductions of
distributive lattices.
\begin{proposition}
\label{prop:upwards families}
  If $\cF$ is an upwards graph family with an $\epsilon$-error randomized
  universal graph sequence $U = (U_i)$ satisfying $|V(U_i)| \leq c$ for some
  constant $c$ (which may depend on $\epsilon$), then there exists a graph $U^*$
  of size $c$ such that $\forall G \in \cF, G \rembed_\epsilon U^*$.
  Furthermore, for any $i < j$ and any $G \in \cF_i$, there exists $G' \in
  \cF_j$ with $\epsilon$-error embedding $g' : V(G') \to V(U^*)$ such that $G$
  is an induced subgraph of $G'$ and the restriction of $g'$ to the domain
  $V(G)$ is an $\epsilon$-error embedding $V(G) \to V(U^*)$.
\end{proposition}
\begin{proof}
  Let $G \in \cF_i$ and let $G' \in \cF_{i+1}$ be such that $G$ is an induced
  subgraph of $G'$. Let $g' : V(G') \to V(U_{i+1})$ the random function
  determined by the randomized universal graph sequence. Then $g'$ restricted to
  the domain $V(G) \subset V(G')$ satisfies
  \[
    \Pru{g'}{ U_{i+1}(g'(x),g'(y)) = G(x,y) }
    = \Pru{g'}{ U_{i+1}(g'(x),g'(y)) = G'(x,y) }
    > 1-\epsilon \,.
  \]
  Therefore we may replace $U_i$ with $U_{i+1}$ in the sequence, for any $i$.

  Since each $U_i$ has size at most $c$, there are at most $2^{c^2}$ graphs
  $U_i$ appearing in the sequence $U$. Thus there is some graph $U^*$ that
  occurs an infinite number of times in the sequence. For every $i$ there exists
  $j > i$ such that $U_j = U^*$.  By applying the above argument, we may replace
  $U_i$ with $U_j = U^*$ in the sequence. We arrive at the sequence $U' =
  (U'_i)$ with $U'_i = U^*$ for every $i$.
\end{proof}

\subsection{Implicit Graph Representations and Induced-Universal Graphs}
\label{subsection:implicit graphs}

Kannan, Naor, and Rudich \cite{KNR92} call a family of graphs an \emph{implicit}
graph family if each of the $n$ vertices can be given a label of $O(\log n)$
bits so that adjacency can be determined from the labels of two vertices.
They observe that an implicit encoding gives an upper bound on the size of an
\emph{induced-universal graph}.  We define these terms below in slightly more
generality (and omit the requirement that encoding and decoding be done in
polynomial time):
\begin{definition}
  Let $\cF = (\cF_i)$ be a graph family and $m(i)$ a function of the graph size.
  The family $\cF$ has an \emph{$m$-implicit} encoding if $\forall i, \exists
  F_i : \zo^{m(i)} \times \zo^{m(i)} \to \zo$ such that $F_i$ is symmetric and
  $\forall G \in \cF_i, \exists g : V(G) \to \zo^{m(i)}$ satisfying $\forall x,y
  \in V(G), F_i(g_i(x),g_i(y)) = G(x,y)$.

  For a graph family $\cF = (\cF_i)$, an \emph{induced-universal graph sequence}
  is a sequence $U = (U_i)$ such that for each $i$ and all $G \in \cF_i$, $G$ is
  an induced sugraph of $U_i$.
\end{definition}
Our notion of $\embed$-universal graphs differs from induced-universal graphs,
since the embedding relation $G \embed U_i$ allows non-injective mappings (two
vertices of $G$ may be mapped to the same vertex in $U_i$). This difference
accounts for the extra factor $n(i)$ in the next theorem.
\begin{theorem}[\cite{Spin03}]
  \label{thm:implicit to universal}
  Let $\cF = (\cF_i)$ be a graph family with size $n(i)$. If there exists an
  $m$-implicit encoding of $\cF$ there is an induced-universal graph sequence $U
  = (U_i)$ such that $|U_i| \leq n(i)2^{m(i)} = 2^{m(i)+\log n(i)}$.
\end{theorem}
Due to the fact that a deterministic universal SMP protocol may always be
assumed to be symmetric (Proposition \ref{prop:deterministic protocols are
symmetric}), it follows by definition and from Lemma \ref{lemma:randomized to
deterministic} that:

\begin{theoremref}[\ref{thm:universal smp to adjacency labeling}]
A graph family $\cF = (\cF_i)$ is $m$-implicit iff $D^\ob(\cF_i) \leq m(i)$ for
every $i$. Therefore, $\cF$ is $O(R^\ob(\cF) \cdot \log n)$-implicit.
\end{theoremref}

If one's goal is merely to obtain an $O(1)$-cost universal SMP protocol for a
family $\cF$, the next observation shows that it suffices to find an
$O(1)$-cost, public-coin, 2-way protocol for each member of $\cF$. Therefore
the family of all graphs with an $O(1)$-cost 2-way protocol is an implicit graph
family with a polynomial-size induced-universal graph. 

\begin{corollaryref}[\ref{cor:2-way to implicit}]
Let $\cF = (\cF_i)$ be a family of graphs with size $n(i)$ and suppose
that for every graph $G \in \cF_i$ there is an $\epsilon$-error 2-way randomized
communication protocol with cost at most $c(i)$. Then $R^\ob_\epsilon(\cF) \leq
2^{c(i)}$. Furthermore, for any fixed constant $c$, the family $\cF$ of graphs
with $R^\leftrightarrow(\ADJ(G)) \leq c$ is $O(\log n)$-implicit.
\end{corollaryref}
\begin{proof}
  Every 2-way, deterministic cost $c$ protocol can be represented as a binary
  tree with at most $2^c$ nodes, where each node is owned by either Alice or Bob
  and the message sent at each step is a 0 or 1 informing the other player of
  which branch to take in the tree. A randomized 2-way protocol is a
  distribution over such trees. To obtain a universal SMP protocol for the
  family $\cF$, Alice and Bob do the following. On input $G \in \cF$ and $x,y
  \in V(G)$, Alice and Bob use shared randomness to draw the deterministic cost
  $c$ protocol for $G$ from the distribution defined by the randomized 2-way
  protocol. Alice sends the size $2^c$ protocol tree and for each node she owns
  she identifies the branch to be taken. Bob does the same. The referee may then
  simulate the protocol.  The conclusion follows from Theorem
  \ref{thm:universal smp to adjacency labeling}.
\end{proof}

\section{Distance Labeling of Distributive Lattices}
\label{section:distributive lattices}

Distributive lattices and distances on these lattices will be defined in the
next subsection, where we also give a necessary lemma characterizing the
distances in terms of the \emph{meet} and \emph{join}.  We will then present an
$O(k \log k)$ weakly-universal protocol and an $O(k^2)$ universal communication
protocol for the family $\cD^k$, where $\cD$ are the distributive lattices. This
implies a $O(k^2 \log n)$-implicit encoding $\cD^k$ of the family $\cD$ of
distributive lattices. The $O(k \log k)$ weakly-universal protocol is optimal
for sufficiently small values of $k$, since it applies to the $k$-Hamming
Distance problem as a special case, for which Sa\v{g}lam \cite{Sag18} recently
gave a matching lower bound (even for 2-way communication). We obtain this
result by adapting the optimal $O(k \log k)$ communication protocol for
$k$-Hamming Distance originally presented by Huang \etal~\cite{HSZZ06}.  

\ignore{
It is easy to
check that $\Omega(\log n)$ bits is required for any adjacency labeling scheme.
\begin{proposition}
  Any $m$-implicit representation for $\cD$ must have $m(n) \geq \log n$.
\end{proposition}
\begin{proof}
  The Boolean hypercube $H_d$ of size $n = 2^d$ is a distributive lattice, and
  it has no two vertices $u \neq v$ such that $u \equiv v$. Therefore if $H^d
  \embed U$ for any graph $U$, we must have $|U| \geq 2^d$. So $D^\ob(\cD_n) \geq
  \ceil{\log |H_d^\equiv|} = \ceil{\log n}$.
\end{proof}
}
We also consider \emph{modular} lattices, a generalization of distributive
lattices, and show that deciding $\dist(x,y) \leq 2$ requires a protocol with
cost $\Omega(n^{1/4})$.

\subsection{Preliminaries on Distributive Lattices}
A lattice is a type of partial order.
We briefly review distributive lattices (see e.g.~\cite{CLM12} for a good
introduction) and then give a characterization of distances in
modular and distributive lattices. The undirected graphs we study are the
\emph{cover graphs} of partial orders. For $x,y$ in a partial order $P$, we say
that $y$ \emph{covers} $x$ and write $x \prec y$ if $\forall z \in P$: if $x
\leq z < y$ then $x=z$.  The \emph{cover graph} (which is the undirected version
of the \emph{transitive reduction}) is the graph $\cov(P)$ on vertex set $P$
with an edge $\{x,y\}$ iff $x \prec y$ or $y \prec x$.

We will define a few types of lattices.
\begin{definition}
  Let $(P,<)$ be a partial order. For a pair $x,y \in P$:
  \begin{itemize}
    \item If the set $\{ z \in P : x,y \geq z \}$ has a unique maximum, we call
      that maximum the \emph{join} of $x,y$ and write it as $x \wedge y$;
    \item If the set $\{ z \in P : x,y \leq z \}$ has a unique minimum, we call
      that minimum the \emph{meet} of $x,y$ and write it as $x \vee y$.
  \end{itemize}
  If $\forall x,y \in P$ the elements $x \wedge y, x \vee y$ exist, then $P$ is
  a \emph{lattice}. A lattice $L$ is \emph{ranked} if there exists a rank function
  such that $x \prec y \implies \rank(x) + 1 = \rank(y)$ and the
  minimum element $0_L$ satisfies $\rank(0_L) = 0$.
  A finite lattice $L$ is \emph{upper-semimodular} if for every $x,y \in L$,
  $x \wedge y \prec x,y \implies x,y \prec x \vee y$. $L$ is
  \emph{lower-semimodular} if for every $x,y \in L$,
  $x,y \prec x \vee y \implies x \wedge y \prec x,y$. $L$ is \emph{modular} if
  it is both upper- and lower-semimodular.
  A lattice $L$ is \emph{distributive} if for all $x,y,z \in L$,
  $x \wedge (y \vee z) = (x \wedge y) \vee (x \wedge z)$. Every distributive
  lattice is modular and every modular lattice is ranked \cite{CLM12}.

  A point $x$ in a lattice $L$ is \emph{join-irreducible} if there is no set $S
  \subseteq L$ such that $x = \bigvee S$ and \emph{meet-irreducible} if there is
  no set $S \subseteq L$ such that $x = \bigwedge S$. Write $J(L)$ for the set
  of join-irreducible elements.

  A subset $D$ of a partial order $P$ is a \emph{downset} or \emph{ideal} if:
  for all $x,y \in L$, if $x \in D$ and $y \leq x$ then $y \in D$. We will write
  $D(P)$ for the set of ideals of $P$.
\end{definition}
\begin{theorem}[Birkhoff (see e.g.~\cite{CLM12})]
  Every distributive lattice $L$ is isomorphic to the lattice of downsets of the
  partial order on its join-irreducible elements, ordered by inclusion; i.e.~$L
  \simeq D(J(L))$, with the meet and join operations given by set union and
  intersection respectively.
\end{theorem}
We need to prove some facts about distances in modular lattices.
\begin{proposition}
  Let $L$ be a graded lattice and let $x,y \in L$. Then $\dist(x,y) \geq
  |\rank(x) - \rank(y)|$, with equality if $x < y$ or $y < x$.
\end{proposition}
\begin{proof}
  This follows from the fact that for every edge $u \prec v$ in the path from
  $x$ to $y$ has $\rank(u)+1 = \rank(v)$.
\end{proof}
To prove our characterization of distance, we define \emph{inversions} in the
path.
\begin{definition}
  Let $L$ be a lattice and let $c_1, \dotsc, c_m$ be a path in $\cov(L)$, so
  that $c_i \prec c_{i+1}$ or $c_{i+1} \prec c_i$ for each $i$. If $c_{i-1},
  c_{i+1} \prec c_i$ or $c_i \prec c_{i-1},c_{i+1}$ we call $c_i$ an
  \emph{inversion} on the path.
\end{definition}
\begin{lemma}
\label{lemma:distance in modular lattices}
The following holds for any $x,y$ in a lattice $\cM$:
\begin{enumerate}
  \item If $\cM$ is lower-semimodular then $\dist(x,y) = \dist(x, x \wedge y) +
    \dist(y, x \wedge y)$;
  \item If $\cM$ is upper-semimodular then $\dist(x,y) = \dist(x, x \vee y) +
    \dist(y, x \vee y)$;
  \item If $\cM$ is distributive then $\dist(x,y) = |X \Delta Y|$ where $X,Y \in
    D(J(\cM))$ are isomorphic images of $x,y$ in Birkhoff's representation.
\end{enumerate}
\end{lemma}
\begin{proof}
  It suffices to prove the first statement: the second follows by the analogous
  argument and the third follows from the modulartiy of distributive lattices
  and Birkhoff's representation.

  Let $\cM$ be lower-semimodular, let $x,y \in \cM$, and let $x = c_0, c_1,
  \dotsc, c_m = y$ be a shortest path between $x$ and $y$, so that $\dist(x,y) =
  \dist(x,c_i) + \dist(y,c_i)$ for any $i$. The statement holds trivially when
  $x < y$ or $y < x$ (since $x \wedge y = x$ or $x \wedge y = y$), so we assume
  $x,y$ are incomparable. We prove the statement by induction on the largest
  rank of an inversion of the form $c_{i-1},c_{i+1} \prec c_i$ in the path.

  First suppose that $c_i$ is any element of the path and assume for
  contradiciton that $\rank(c_i) < \rank(x \wedge y)$. Then
  \[
  \dist(x, x \wedge y) = \rank(x) - \rank(x \wedge y) 
  < \rank(x) - \rank(c_i) \leq \dist(x,c_i),
  \]
  a contradiction. Thus $\rank(c_i) \geq \rank(x \wedge y)$ for each element of
  the path.

  Suppose there are no inversions of the form $c_{i-1},c_{i+1} \prec c_i$. Then
  $c_i < x,y$ and therefore $c_i \leq x \wedge y$ so $\rank(c_i) \leq \rank(x
  \wedge y)$, and by the above inequality we have $\rank(c_i) \geq \rank(x
  \wedge y)$, so $\rank(c_i) = \rank(x \wedge y)$. Therefore, as desired,
  \begin{align*}
  \dist(x,y)
  &= \dist(x,c_i) + \dist(y,c_i)
  = \rank(x) - \rank(c_i) + \rank(y) - \rank(c_i) \\
  &= \rank(x) - \rank(x \wedge y) + \rank(y) - \rank(x \wedge y) \\
  &= \dist(x,x \wedge y) + \dist(y,x\wedge y) \,.
  \end{align*}
  Now let $c_i$ be an inversion of the form $c_{i-1},c_{i+1} \prec c_i$ with
  $\rank(c_i) > \rank(x \wedge y)$. Then by lower-semimodulariity there is an
  element $c'_i = c_{i-1} \wedge c_{i+1} \prec c_{i-1}, c_{i+1}$. Then replacing
  $c_i$ with $c'_i$ maintains the length of the path. Performing the same
  operation on all such inversions of maximum rank reduces the maximum rank by 1
  and the result holds by induction.
\end{proof}

\subsection{A Universal Protocol for Distributive Lattices}
Write $\cD = (\cD_n)$ for the family of cover graphs of distributive lattices on
$n$ vertices.  We first give an optimal protocol for distances in distributive
lattices in the \emph{weak} universal model (recall that in this model, the
referee sees the shared randomness). This protocol is adapted from a simplified
presentation of Huang \etal's $k$-Hamming Distance protocol (\cite{HSZZ06})
communicated to us by E.~Blais.
\begin{theorem}
  For any $\epsilon > 0$ and integer $k$, $R^\weak_\epsilon(\cD^k) = 
  O\left(k\log(k/\epsilon)\right)$.
\end{theorem}
\begin{proof}
  For any distributive lattice $L \simeq D(J(L))$, identify each vertex $x \in
  L$ with its ideal $X \subseteq J(L)$ of join-irreducibles. Write $e_1, \dotsc,
  e_m$ for the basis vectors of $\bF_2^m$.
  Consider the following protocol. On the distributive lattice $L$ and vertices
  $x,y$, Alice and Bob perform the following:
\begin{enumerate}
  \item Define $m = \ceil{\frac{(k+2)^2}{\epsilon}}, q =
    \ceil{\log\frac{1}{\epsilon} + \log \sum_{i=0}{m \choose i}}$.
  \item Let $S = (s_1, \dotsc, s_m)$ be a multiset of uniformly random elements
    of $\bF_2^q$.
  \item For each join-irreducible element $j \in J(L)$ assign a uniformly random
    index $i_j \sim [m]$.
  \item For each vertex $v \subseteq J(L)$ there is an indicator vector $a(v) \in \bF_2^m$
    defined by $a(v) = \sum_{j \in v} e_{i_j}$. Label $v$ with
    $\ell(v) = \sum_{i=1}^m a(v)_i s_i$.
  \item Alice sends $\ell(x)$ and Bob sends $\ell(y)$ to the referee.
  \item The referee accepts iff $\ell(x) + \ell(y)$ is a sum of at most $k$
    elements of $S$.
\end{enumerate}
By Lemma \ref{lemma:distance in modular lattices} and Birkhoff's theorem,
$\dist(x,y) = \dist(x,x\wedge y) + \dist(x \wedge y,y) = |X \setminus Y| + |Y
\setminus X| = |X \Delta Y|$, where $\Delta$ denotes the symmetric difference.
Suppose $\dist(x,y) = |X \Delta Y| \leq k$. Then $\ell(x) + \ell(y) = \sum_{j
\in X \Delta Y} c(j)$ is a sum of at most $k$ elements of $S$, so the protocol
accepts with probability 1 (so this protocol has 1-sided error).

Now suppose $\dist(x,y) = |X \Delta Y| \geq k+1$. The correctness of the
protocol follows from the next two claims along with the observations that
$a(x) + a(y) = a(x \wedge y)$ and
$\ell(x) + \ell(y) = \ell(x \wedge y)$ (with arithmetic in $\bF_2$) and that
$\dist(x,y) \geq k+1$ implies $\rank(x \wedge y) \geq k+1$. We will write
$|a(v)|$ for the number of 1's in the vector $a(v)$.
\begin{claim}
\label{claim:indicator has large weight}
  Any vertex $v \subseteq J(L)$ with $\rank(v) \geq k + 1$ has $|a(v)| \geq k+1$
  with probability at least $1-\epsilon/2$.
\end{claim}
\begin{proof}[Proof of claim]
  If $\rank(v) = k+1$, so $v$ is a
  set of $k+1$ join-irreducibles, then the probability that any two indices
  $i_j,i_{j'}$ collide, for $j,j' \in v$, is by the union bound at most
  \[
    {k+1 \choose 2} \Pr{i_j = i_{j'}}
    = \frac{k(k+1)}{2} \frac{1}{m}
    \leq \frac{(k+1)^2}{2} \frac{\epsilon}{(k+2)^2} = \epsilon / 2 \,.
  \]
  For $\rank(v) > k+1$ choose $v' \prec v$ so $k+1 \leq \rank(v') < \rank(v)$,
  so using induction and the assumption $\epsilon < 1/2$,
  \begin{align*}
    \Pr{|a(v)| \leq k}
    &= \frac{k+1}{m} \Pr{|a(v')| = k+1} + \frac{k}{m}\Pr{|a(v')| \leq k}
    < \frac{\epsilon}{k+2} + \frac{\epsilon}{k+2}\cdot \frac{\epsilon}{2} \\
    &= \epsilon \left( \frac{1}{k+2} + \frac{\epsilon}{2(k+2)} \right)
    \leq \epsilon \left(\frac{1}{3} + \frac{1}{12}\right) < \epsilon/2 \,.
    \qedhere
  \end{align*}
\end{proof}
\begin{claim}
  For any vertex $v \subseteq J(L)$, if the indicator vector $a(v)$ has weight
  $\geq k+1$ then, with probability at least $1-\epsilon/2$, $\ell(v)$ is not a
  sum of at most $k$ vectors in $S$.
\end{claim}
\begin{proof}[Proof of claim]
  Write $kS$ for the set of all sums of at most $k$ vectors of $S$.
  Fix any $a(v)$ with weight $\geq k+1$ and let $A = \{ i : a(v)_i = 1 \}$ so
  $|A| \geq k+1$. Let $b \in kS$ be any sum of $k$ vectors in $S$, and let $B \subset
  [m]$ be a set of indices of size $|B| \leq k$ such that $b = \sum_{i \in B}
  s_i$.

  Since $|B| \leq k < |A|$ we must always have $A \setminus B \neq \emptyset$
  and $\ell(v) + b = \sum_{i \in A \setminus B} s_i$, so $\Pr{\ell(v) + b = 0} =
  2^{-q}$. Therefore, by the union bound over all such vectors $b$,
  \[
    \Pr{\ell(v) \in kS}
    \leq \sum_{i=0}^k {m \choose i} 2^{-q}
    < \epsilon/2 \,. \qedhere
  \]
\end{proof}
We can put a bound on $q$ by using
\[
  \sum_{i=0}^k {m \choose i} \leq k{m \choose k} \leq k
  \left(\frac{em}{k}\right)^k
\]
so
\[
q \leq 1 + \log\frac{1}{\epsilon} + \log k + k \log \frac{em}{k}
\leq \log\frac{2k}{\epsilon} + k \log\ceil{\frac{ek}{\epsilon}}
= O\left(k \log \frac{k}{\epsilon} \right) \,. \qedhere
\]
\end{proof}
Observe that the referee must see the set $S$ for the above protocol to work.
We can easily modify the above protocol to get $O(k^2)$.
\begin{theorem}
  \label{thm:distributive lattice protocol}
  For any $\epsilon > 0$ and any integer $k$, $R^\ob_\epsilon(\cD^k) =
  O\left(k^2\log(1/\epsilon)\right)$.
\end{theorem}
\begin{proof}
  The protocol is the same as above, with the following modification: Alice and
  Bob each send the indicator vectors $a(x),a(y) \in \bF_2^m$.

The correctness of this protocol for error $1/3$ follows from Claim
\ref{claim:indicator has large weight}. Observe that Alice and Bob use the same
strategy to send their messages and that the decision function is symmetric. The
communication cost is now at most $m = \ceil{3(k+2)^2/2}$.

This protocol is one-sided, so to achieve error $\epsilon$ we can run the
protocol $r = \ceil{\log_3 (1/\epsilon)}$ times and take the AND of the results. The
probability of failure is $(1/3)^r = 3^{-r} < \epsilon$.
\end{proof}
Now we apply Theorem \ref{thm:universal smp to adjacency labeling} to obtain
Theorem \ref{thm:distributive lattices}.

Since the family of distributive lattices is an upwards family (simply append a
new least element to obtain a larger distributive lattice), we see from
Proposition \ref{prop:upwards families} that lattices in $\cD^k$ can be randomly
embedded into a constant-size graph, for any constant $k$.  In fact, by
inspection of the protocol, we see that the family $\cD$ can be randomly
embedded into a small-dimensional hypercube, while $\cD^k$ can be embedded into
the $k$-closure of the $O(k^2)$-dimensional hypercube. 
\begin{corollary}
  For any $\epsilon > 0$ and any $k$, there exists a graph $U$ of size
  $2^{O(k^2\log(1/\epsilon))}$ such that for all $L \in \cD^k$, $L
  \rembed_\epsilon U$.
\end{corollary}

\subsection{Lower Bound for Modular Lattices}
Since Lemma \ref{lemma:distance in modular lattices} works for any modular
lattices, it is natural to ask whether we can achieve a similar constant-cost
protocol for computing distance thresholds in modular lattices. However, we show
that this is impossible.
\begin{lemma}
  There is a function $m(n) = O(n^4)$ such that if $G$ is any graph with $n$
  vertices (where $G(u,u) = 1$ for all $u$), there exists a modular lattice $M$
  with size $m(n)$ such that $G$ is an induced subgraph of $\cov(M)^2$.
\end{lemma}
\begin{proof}
  Construct the lattice $M$ as follows:
  \begin{enumerate}
    \item Start with vertices $V$, which are all incomparable.
    \item For each edge $e = \{u,v\} \in E$, add vertices $a_e,b_e$ such that
      $a_e < u,v < b_e$.
    \item $\forall e = \{u,v\},e' = \{u',v'\} \in E$ such that $e \cap e' = \emptyset$ add a vertex
      $c_{e,e'}$ with $a_e,a_{e'} < c_{e,e'} < b_e,b_{e'}$.
    \item Add vertices $0_M$ and $1_M$ such that $0_M < a_e$ and $b_e < 1_M$ for
      all $e \in E$.
  \end{enumerate}
  First we prove that $M$ is a modular lattice and then we prove the bound on
  the size.
  \begin{claim}
    $M$ is a modular lattice.
  \end{claim}
  \begin{proof}[Proof of claim]
  Observe that all orderings $<$ directly imposed by this process are covering
  orders $\prec$. Let $A = \{a_e\}_{e \in E}, B = \{b_e\}_{e \in E}, C =
  \{c_e\}_{e \in E}$ and $V$ the
  original set of vertices. By construction, $M$ is graded with $\rank(0_M) = 0,
  \rank(A) = 1, \rank(V) = \rank(C) = 2, \rank(B) = 3, \rank(1_M) = 4$.
  Note that for every pair of vertices $x,y \in M, 0_M \leq x,y \leq 1_M$ so
  upper- and lower-bounds exist.

  Assume for contradiction that $M$ is not a modular lattice, so there exist
  incomparable $x,y \in M$ such that either $x\wedge y$ or $x \vee y$ does not
  exist, or such that $x \wedge y \prec x,y \not \prec x \vee y$ or $x \wedge y
  \not \prec x,y \prec x \vee y$.

  Case 1: Suppose $\rank(x) \neq \rank(y)$. Then $x \wedge y = 0_M$ and $x \vee
  y = 1_M$ so $x \wedge y \not \prec x,y \not \prec x \vee y$.

  Case 2: Suppose $x,y \in A$ so $x = a_e,y=a_{e'}$. Then $0_M = a_e \wedge a_{e'} \prec
  a_e,a_{e'}$. If $a_e,a_{e'} < u,v$ for $u,v \in V$ then $u,v \in e \cap {e'}$ so
  $u=v$. If $a_e,a_{e'} < v, c_{d,d'}$ for $v \in V$ and $c_{d,d} \in C$ then $v
  \in e \cap {e'}$ and $c_{d,d'} = c_{e,e'}$ so $e \cap {e'} = \emptyset$, a
  contradiction. Finally, if $a_e,a_{e'} < c_{d,d'}, c_{d',d''}$ then $c_{d,d'}
  = c_{d',d''} = c_{e,e'}$. So $a_e \vee a_{e'}$ exists and $a_e \wedge a_{e'} \prec
  a_e,a_{e'} \prec a_e \vee a_{e'}$. The same argument holds for $x,y \in B$.

  Case 3: Suppose $x,y \in V$ and assume $a_e,a_{e'} < x,y$. Then $x,y \in e \cap
  {e'}$ so $a_e = a_{e'}$. A similar argument holds for $x,y < b_e,b_{e'}$. So $x
  \wedge y \prec x,y \prec x \vee y$.

  Case 4: Suppose $x,y \in C$ so $x = c_{e,e'}, y = c_{d,d'}$. Suppose
  $a_s,a_t < c_{e,e'}, c_{d,d'}$. Then $s,t \in \{e,e'\} \cap \{d,d'\}$ so
  either $\{e,e'\} = \{d,d'\}$ or $s = t$. The same argument holds for
  $c_{e,e'},c_{d,d'} < b_s,b_t$ so $x \wedge y \prec x,y \prec x \vee y$.

  Case 5: Suppose $x \in V, y \in C$ so $y = c_{e,e'}$ which implies $e \cap
  {e'} = \emptyset$. If $x \notin e \cup {e'}$ then $x \wedge c_{e,e'} = 0_M$ and
  $x \vee c_{e,e'} = 1_M$ so $x \wedge c_{e,e'} \not \prec x,c_{e,e'} \not \prec x
  \vee c_{e,e'}$; so suppose $x \in e \cup {e'}$. If $a_e,a_{e'} < x,c_{e,e'}$ then
  $x \in e \cap {e'}$ which is a contradiction. Then $x \in e$ or $x \in
  {e'}$; say $x \in e$. Then $a_e = x \wedge c_{e,e'}$. The same argument holds
  for $B$ so $a_e = x \wedge c_{e,e'} \prec x,c_{e,e'} \prec x \vee c_{e,e'} =
  b_e$.
  \end{proof}
  \begin{claim}
    $G$ is an induced subgraph of $\cov(M)^2$.
  \end{claim}
  \begin{proof}[Proof of claim]
    Suppose $\{u,v\} \in E$. Then there is $a_e \prec u,v$ so $\dist(u,v) \leq
    2$ in $\cov(M)$.
    Now let $u,v \in V(G)$ and suppose $\dist(u,v) \leq 2$ in $\cov(M)$ so that,
    by Lemma \ref{lemma:distance in modular lattices},
    $u \wedge v \prec u,v \prec u \vee w$. By construction, either $u=v$ so
    $G(u,v) = G(u,u) = 1$, or $u \wedge v = a_e$ for some $e \in E(G)$ so $u,v
    \in e$ and therefore $G(u,v) = 1$.
  \end{proof}
  The size of $M$ is at most $2 + |E(G)| + |E(G)|^2 = O(n^4)$.
  Let $m(n)$ be the maximum size of a modular lattice obtained in this way from
  a graph of size $n$. We want all constructions to be of the same size, so
  repeatedly append new least elements until the size reaches $m(n)$; this
  maintains the modular lattice property.
\end{proof}

\begin{theorem}
\label{thm:lower bound for modular lattices}
  Let $\cM = (\cM_n)$ be the family of cover graphs of modular lattices.
  $R^\ob(\cM^2) \geq \Omega(n^{1/4})$.
\end{theorem}
\begin{proof}
  Suppose there is a protocol for $\cM^2$ with cost $o(n^{1/4})$.
  Given a graph $G$ of size $n$, Alice and Bob construct the modular lattice of
  size $m(n) = O(n^4)$ with $G$ an induced subgraph of $\cov(M)^2$ and run the
  protocol for $\cM^2$ with size $m(n)$ (observe that all possible constructions
  must be of the same size, since the referee does not know which lattice Alice
  and Bob construct). This has cost $o(m(n)^{1/4}) = o(n)$, which contradicts
  Theorem \ref{thm:lower bound for all graphs}.
\end{proof}

\section{Communication on Efficiently Labelable Graphs}
\label{section:planar graphs}

In this section we take inspiration from the field of implicit graphs and graph
labeling and show that one may often, but not always, obtain constant-cost
adjacency and $k$-distance protocols for families that are commonly studied in
the graph labeling literature. 

\subsection{Trees, Forests, and Interval Graphs}

In this section we pick the low-hanging fruit from trees and forests (and
interval graphs). Applying Theorem \ref{thm:universal smp to adjacency labeling}
with the next lemma, we get Theorem \ref{thm:k distance for trees}.

\begin{lemma}
Let $\cT = (\cT_n)$ be the family of trees of size $n$.  $R^\ob_\epsilon(\cT^k)
= O\left(k \log \frac{1}{\epsilon} \right)$, and this protocol will correctly
compute the distance in the case $\dist(x,y) \leq k$.
\end{lemma}
\begin{proof}
  Consider the following protocol. On input $(T,x),(T,y)$ for a tree $T$, Alice
  and Bob perform the following.
  \begin{enumerate}
    \item Partition the vertices of $T$ into sets $T_1, \dotsc, T_m$ such that
      $T_i = \{ v \in V(T) : (i-1)k \leq \depth(v) < ik \}$. For each $v \in
      V(T)$ let $t(v)$ be the index of the unique set satisfying $v \in
      T_{t(v)}$.
    \item For each vertex $v \in V(T)$ assign a uniformly random color $\ell(v)$
      in $[m]$ for $m = \ceil{6/\epsilon}$. Let $x'$ be root of the subtree of
      $T_{t(x)}$ that contains $x$, and let $x''$ be the root of the subtree of
      $T_{t(x)-1}$ that contains $x$. Let $x_0, x_1, \dotsc, x_k, \dotsc,
      x_{k_1} = x$ be the path from $x''$ to $x$ (with $x_k = x'$) and let $y_0,
      \dotsc, y_k, \dotsc, y_{k_2}$ be the path from $y''$ to $y$. Alice and Bob
      send $\ell(x_0), \dotsc, \ell(x_{k_1})$ and $\ell(y_0), \dotsc,
      \ell(y_{k_2})$ respectively.
    \item If $\ell(x') = \ell(y')$, let $p$ be the maximum index such that
      $\ell(x_i) = \ell(y_i)$ for each $k < i \leq p$. Let $d =
      (k_1-p)+(k_2-p)$. If $\ell(x'') = \ell(y'')$, let $p$ be the maximum
      index such that $\ell(x_i) = \ell(y_i)$ for each $i \leq p$ and let $d =
      (k_1-p)+(k_2-p)$. If $\ell(x'') = \ell(y')$ let $p$ be the maximum index such
      that $\ell(x_i) = \ell(y_{k+i})$ for each $i \leq p$ and let $d = (k_1-p)
      + (k_2-k-p)$. If $\ell(x') = \ell(y'')$ do the same with $x,y$ reversed.
      In each case, if $d \leq k$, the referee outputs $d$, otherwise they
      output ``$>k$''. If none of the above cases hold, output ``$>k$''.
  \end{enumerate}
  The cost of this protocol is $2k\ceil{\log m} = O(k\log(1/\epsilon))$.
  With probability at least $1-4/m > 1-\epsilon/2$, each of the possible
  equalities $x''=y'',x'=y',x''=y',x'=y''$ will be correctly observed by the
  referee.  If $\{x',x''\} \cap \{y',y''\} = \emptyset$ then $x,y$ are not in
  the same subtree rooted at depth $\depth(x'')$, so the distance from $x$ to
  any common ancestor of $x,y$ is at least $\dist(x,x'') > k$. Therefore if
  $\dist(x,y) \leq k$, one of these equalities will hold. If $x'' = y''$ and $q$
  is the maximum integer such that $x_i = y_i$ for all $i \leq q$ then
  $\dist(x,y) = (k_1-q)+(k_2-q)$, because the deepest common ancestor of $x,y$
  is at depth $\depth(x_0) + q$. Conditional on the 4 equalities being correctly
  observed, we will have $d = (k_1-p)+(k_2-p) \leq k$ since $p \geq q$. If $p >
  q$ then $\ell(x_{q+1})=\ell(y_{q+1})$ even though $x_{q+1} \neq y_{q+1}$,
  which occurs with probability $1/m < \epsilon/2$. Therefore the probability
  that $d \neq \dist(x,y)$ is at most $2(\epsilon/2) =\epsilon$ when $\dist(x,y)
  \leq k$. A similar argument holds in the other 3 cases.

  If $\dist(x,y) > k$ then still with probability at least $1-\epsilon/2$ all 4
  possible equalities are correctly observed. Following the same argument as in
  the equality case, we see that if any of the equalities hold we will have $d =
  \dist(x,y)$ with probability greater than $1-\epsilon/2$, for total error
  probability less that $\epsilon$. If none of the 4 equalities hold then the
  probability of error is at most $\epsilon/2$.
\end{proof}
Since trees have efficient protocols, one might wonder about generalizations of
trees. The \emph{arboricity} of a graph is one such generalization, which
measures the minimum number of forests required to partition all the edges.
\begin{definition}
  A graph $G = (V,E)$ has \emph{arboricity} $\alpha$ iff there exists an edge
  partition of $G$ into forests $T_1, \dotsc, T_\alpha$. Equivalently, for $S$
  ranging over the set of subgraphs of $G$, $G$ has
  \[
      \max_{S} \left\lceil \frac{E(S)}{V(S)-1} \right\rceil \leq \alpha \,.
  \]
\end{definition}
Low-arboricity graphs easily admit an efficient universal SMP protocol for
adjacency.
\begin{proposition}
\label{prop:low arboricity protocol}
  Let $\cF$ be any family of graphs with arboricity at most $\alpha$.
  For all $\epsilon > 0, R^\ob_\epsilon(\cF) = O\left(\alpha
  \log\frac{\alpha}{\epsilon}\right)$.
\end{proposition}
\begin{proof}
  On the graph $G$ and vertices $x,y$, Alice and Bob perform the following:
  \begin{enumerate}
    \item Compute a partition of $G$ into $\alpha$ forests $T_1, \dotsc,
    T_\alpha$.
    \item Assign to each vertex $v$ a uniformly random number $\ell(v) \sim 
    [m]$ for $m = \ceil{2\alpha/\epsilon}$.
    \item Let $x_i$ be the parent of $x$ in tree $i$ and let $y_i$ be the parent
    of $y$. Alice sends $\ell(x)$ and $\ell(x_i)$ for each $i$, and Bob does
    this same with $y$.
    \item The referee accepts iff $\ell(x) = \ell(y_i)$ or $\ell(y) = \ell(x_i)$
    for any $i$.
  \end{enumerate}
  This protocol has one-sided error since if $x,y$ are adjacent then either $x_i
  = y$ or $y_i = x$ for some $i$, so the referee will accept with probability 1.
  If $x,y$ are not adjacent then the referee will accept with probability at
  most $2\alpha \cdot \frac{1}{m} < \epsilon$.
\end{proof}
However, even graphs of arboricity 2 do not admit efficient
protocols or labeling schemes for distance 2, which we can show by embedding an
arbitrary graph of size $\Omega(\sqrt n)$ into the 2-closure of an arboricity 2
graph of size $n$:
\begin{proposition}
\label{prop:lower bound for arboricity 2 graphs}
  Let $\cF$ be the family of arboricity-2 graphs. Then
  $R^\ob(\cF^2) \geq \Omega(\sqrt n)$.
\end{proposition}
\begin{proof}
  The lower bound is obtained via Theorem \ref{thm:lower bound for all graphs}
  in the same way as in Theorem \ref{thm:lower bound for modular lattices},
  using the following construction. 
  For all simple graphs $G = (V,E)$ with $n$ vertices, there exists a graph $A$
  of size $n + {n \choose 2}$ and arboricity 2 such that $G$ is an induced
  subgraph of $A^2$.  Let $A$ be the graph defined as follows:
  \begin{enumerate}
    \item Add each vertex $v \in V$ to $A$;
    \item For each pair of vertices $\{u,v\}$ add a vertex $e_{\{u,v\}}$ and add
      edges $\{u, e_{\{u,v\}}\}, \{v, e_{\{u,v\}}\}$ iff $\{u,v\} \in E$.
  \end{enumerate}
  This graph has arboricity 2 since for each $e_{\{u,v\}}$ we may assign each of
  its 2 incident edges a color in $\{1,2\}$ (if the edges exist). Then the edges
  with color $i \in \{1,2\}$ form a forest with roots in $V$.
\end{proof}
Now we give an example of a family, the interval graphs, with size $O(\log n)$
adjacency labels but with no constant-cost universal SMP protocol; in fact,
randomization does not give more than a constant-factor improvement for this
family.  An \emph{interval graph} of size $n$ is a graph $G$ where for each
vertex $x$ there is an interval $X \subset [2n]$ such that any two vertices
$x,y$ are adjacent in $G$ iff $X \cap Y \neq \emptyset$. These have an $O(\log
n)$ adjacency labeling scheme \cite{KNR92} (one can simply label a vertex with
its two endpoints in $[2n]$).

There is a simple reduction from the \textsc{Greater-Than} communication
problem, in which Alice and Bob receive integers $x,y \in [n]$ and must decide
if $x < y$.  It is known that the one-way public-coin communication cost of
\textsc{Greater-Than} is $\Omega(\log n)$ \cite{MNSW98}, so
$R^\|(\textsc{Greater-Than}) = \Omega(\log n)$.
\begin{proposition}
\label{prop:lower bound for interval graphs}
  For the family $\cF$ of interval graphs, $R^\ob(\cF) = \Omega(\log n)$.
\end{proposition}
\begin{proof}
  We can use a universal SMP protocol for $\cF$ to get a protocol for
  \textsc{Greater-Than} as follows.  Alice and Bob construct the interval graph
  with intervals $[1,i], [i,n]$ for each $i \in [n]$, so there are $2n$ vertices
  in $G$. On input $x,y \in [n]$, Alice and Bob compute adjacency on the
  intervals $[1,x],[1,y]$ and then again on $[1,x],[y,n]$. Assume both runs of
  the protcol succeed. Then when the output is 1 for both runs we must have $y
  \in [1,x]$ so $y \leq x$ and otherwise we have $y \notin [1,x]$ so $x < y$.
\end{proof}

\subsection{Planar Graphs}

Write $\cP_n$ for the set of planar graphs of size $n$ and write $\cP = (\cP_n)$
for the family of planar graphs. Gavoille \etal~\cite{GPPR04} gave an
$O(\sqrt n \log n)$ labeling scheme where $\dist(x,y)$ can be computed from the
labels of $x,y$, and Gawrychowski and Uzna{\'n}ski \cite{GU16} improved this to
$O(\sqrt n)$.  These labeling schemes recursively identify size-$O(\sqrt n)$
sets $S$ and record the distance of each vertex $v$ to each $u \in S$, so the
$\sqrt n$ factor is unavoidable using this technique. We want to solve
$k$-distance with a cost independent of $n$, so we need a new method. Our main
tool is Schnyder's elegant decomposition of planar graphs into trees:

\begin{theorem}[Schnyder \cite{Schn89}, see \cite{Fel12}]
  Define the dimension $\dim(G)$ of a graph $G$ as is the minimum $d$ such that
  there exist total orders $<_1, \dotsc, <_d$ on $V(G)$ satisfying:

  (*) For every edge $\{u,v\} \in E$ and $w \notin \{u,v\}$ there exists $<_i$
  such that $u,v <_i w$.

$G$ is planar iff $\dim(G) \leq 3$. If $G$ is planar then there exists a
partition $T_1,T_2,T_3$ of the edges into directed trees satisfying the
following.  Let $T^{-1}_i$ be edge-induced directed graph on $V(G)$ obtained by
reversing the direction of each edge in $T_i$. The graphs with edges $T_i \cup
T^{-1}_{i-1} \cup T^{-1}_{i+1}$ have linear extensions $<_i$ such that
$<_1,<_2,<_3$ satisfy (*).
\end{theorem}
Schnyder's Theorem implies that the arboricity of planar graphs is at most 3, so
we may use the protocol for low-arboricity graphs (Proposition \ref{prop:low
arboricity protocol}) to determine adjacency in $\cP$, so we move on to $\cP^2$,
which may have large arboricity (arboricity is within a constant factor of
\emph{degeneracy}):
\begin{theorem}[\cite{AH03}]
  There are planar graphs $P$ for which the degeneracy of $P^2$ is $\Theta(\deg
  P)$, where $\deg P$ is the maximum degree of any vertex in $P$.
\end{theorem}
We avoid this blowup in arboricity by treating edges of the form $a \gets b \to
c$ separately (with directions taken from the Schnyder wood). The proof uses the
following split operation:
\newcommand{\ssplit}{\mathsf{split}}
\begin{definition}
  Let $G \in \cP$ and fix a planar map and a Schnyder wood $T_1, T_2, T_3$.
  Define the graph $\ssplit(G)$ by the following procedure (see Figure
  \ref{fig:planar split}):
  \begin{enumerate}
    \item For each vertex $s \in V(G)$ add vertices $s,s_1,s_2,s_3$ to
      $\ssplit(G)$ (excluding $s_i$ if $s$ has no incoming edge in $T_i$). Add
      edges $(s_i,s)$ to $T'_i$;
    \item For each (directed) edge $(u,v) \in T_i$ add the edges
      $(u_{i-1},v_i),(u_{i+1},v_i)$ (arithmetic mod 3) to $T'_i$;
    \item For the unique (directed) edge $(v,u) \in T_i$ add the edges
    $(v_{i-1},u), (v_{i+1},u)$ to $T'_i$.
  \end{enumerate}
\end{definition}
\begin{figure}[h!]
  \center
  \begin{tikzpicture}
    \coordinate (a) at (-2,0);
    \coordinate (b) at (2,0);
    \draw[color=black] (a) circle[radius=0.5em] node {$s$};
    \draw[thick, <-, color=blue, shorten <=0.5em] (a) -- +(-100:1);
    \draw[thick, <-, color=blue, shorten <=0.5em] (a) -- +(-80:1);
    \draw[thick, ->, color=blue, shorten <=0.5em] (a) -- +(90:1);
    \draw[thick, <<-, color=red, shorten <=0.5em] (a) -- +(-220:1);
    \draw[thick, <<-, color=red, shorten <=0.5em] (a) -- +(-200:1);
    \draw[thick, ->>, color=red, shorten <=0.5em] (a) -- +(-30:1);
    \draw[thick, <<<-, color=green, shorten <=0.5em] (a) -- +(20:1);
    \draw[thick, <<<-, color=green, shorten <=0.5em] (a) -- +(40:1);
    \draw[thick, ->>>, color=green, shorten <=0.5em] (a) -- +(210:1);

    \draw[color=black] (b) circle[radius=0.5em] node {$s$};
    \draw[color=black] (b)+(-90:1) circle[radius=0.5em] node {$s_1$};
    \draw[thick, ->, color=blue, shorten <=0.5em, shorten >=0.5em] (b)++(-90:1) -- (b);
    \draw[thick, <-, color=blue, shorten <=0.5em] (b)++(-90:1) -- +(-100:1);
    \draw[thick, <-, color=blue, shorten <=0.5em] (b)++(-90:1) -- +(-80:1);
    \draw[thick, ->, color=blue, shorten <=0.5em, shorten >=0.5em] (b)++(30:1) -- +(135:1);
    \draw[thick, ->, color=blue, shorten <=0.5em, shorten >=0.5em] (b)++(-210:1) -- +(45:1);

    \draw[color=black] (b)+(30:1) circle[radius=0.5em] node {$s_3$};
    \draw[thick, ->>>, color=green, shorten <=0.5em, shorten >=0.5em] (b)++(30:1) -- (b);
    \draw[thick, <<<-, color=green, shorten <=0.5em] (b)++(30:1) -- +(20:1);
    \draw[thick, <<<-, color=green, shorten <=0.5em] (b)++(30:1) -- +(40:1);
    \draw[thick, ->>>, color=green, shorten <=0.5em, shorten >=0.5em] (b)++(-90:1) -- +(165:1);
    \draw[thick, ->>>, color=green, shorten <=0.5em, shorten >=0.5em] (b)++(-210:1) -- +(255:1);

    \draw[color=black] (b)+(-210:1) circle[radius=0.5em] node {$s_2$};
    \draw[thick, ->>, color=red, shorten <=0.5em, shorten >=0.5em] (b)++(-210:1) -- (b);
    \draw[thick, <<-, color=red, shorten <=0.5em] (b)++(-210:1) -- +(-220:1);
    \draw[thick, <<-, color=red, shorten <=0.5em] (b)++(-210:1) -- +(-200:1);
    \draw[thick, ->>, color=red, shorten <=0.5em, shorten >=0.5em] (b)++(-90:1) -- +(15:1);
    \draw[thick, ->>, color=red, shorten <=0.5em, shorten >=0.5em] (b)++(30:1) -- +(285:1);
  \end{tikzpicture}
  \caption{Splitting vertex $s$, with $T_1,T_2,T_3$ in blue, red, and green
  respectively (1,2, and 3 arrowheads).}
  \label{fig:planar split}
\end{figure}
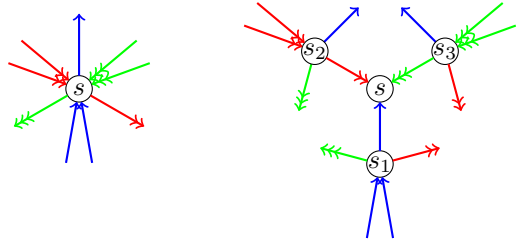
\begin{proposition}
  $\ssplit(G)$ is planar.
\end{proposition}
\begin{proof}
  We prove that splitting any vertex $s$ results in a planar graph. By induction
  we may then split each vertex in sequence and obtain a planar graph. Let
  $<_i$ be any total order on $V(G)$ extending $T_i \cup T^{-1}_{i-1} \cup
  T^{-1}_{i+1}$, which satisfies condition (*) by Schnyder's theorem.
  Let $<'_1, <'_2, <'_3$ be the same total
  orders, extending $T'_1,T'_2,T'_3$, and augmented to include $s_1,s_2,s_3$ as
  follows:
  \begin{enumerate}
    \item For each $u \in V(G)$, $s_i <'_j u$ iff $s <_j u$ and $u <'_j$ iff
  $u <_j s$;
  \item For each $i$, set $s_i <'_i s <'_i s_{i+1} <'_i s_{i-1}$.
  This is possible since $\{s_i\}$ do not have a defined ordering in $<_i$ and
  remain incomparable after the previous step.
  \end{enumerate}
  Note that for any edge $(u,v) \in T'_i$ we have $u <'_i v$ and $v <'_j u$ for
  $j \neq i$.  It suffices to prove that condition (*) is satisfied by the new
  orders.  Let $\{u,v\} \in E(\ssplit(G))$ and let $w \notin \{u,v\}$. We will
  show that there exists $i$ such that $u,v <'_i w$.
  
  If $u,v,w \in V(G)$ then we are done since the orders $<'_i$ are the same as
  $<_i$ on these vertices. 

  If $u = s_i$ then either $v \in V(G) \setminus s$, in which case $v <_i s$ so
  $v <'_i u$ and therefore $u <'_j v$ for $j \neq i$, or $v = s$ so $u <'_i v$
  and therefore $v <'_j u$ for $j \neq i$. Let $v \neq s$. For any $w \in V(G)
  \setminus \{v\}$ we have, by (*), either $v,s <_i w$ so $v <'_i u <'_i s <'_i
  w$, or $v,s <_j w$ so $u <'_j v <'_j w$. If $v = s$ then by construction there
  exists $(u',u) \in T_i$. By (*), either $u',v <_i w$ so $u <'_i v <'_i w$, or
  $u',v <_j w$ so $v <'_j u <'_j u' <'_j w$.

  The only case remaining is if $w = s_i$ and $u,v \in V(G)$. By construction
  there exists $(w',w) \in T_i$. Either $u,v <_i w' <'_i w <'_i s$ or by (*)
  there exists $j$ such that $u,v <_j s$ and since $(w,s)$ is an edge in $T'_i$,
  $s <'_j w$ for $j\neq i$.
\end{proof}

\newcommand{\tofrom}[1]{{\rightarrow #1 \leftarrow}}
\newcommand{\fromto}[1]{{\leftarrow #1 \rightarrow}}
\begin{definition}
  Let $G = (V,E)$ be a planar graph. Fix a planar map and a Schnyder wood $T_1,
  T_2, T_3$.  For each $i$, define the graph $G_i = (V,E \setminus T_i)$ as the
  graph obtained by removing each edge in $T_i$. Define the \emph{head-to-head
  closure} of $G_i$, written $G_i^{\fromto{}}$, as the graph with an edge
  $\{u,v\}$ iff there exists $w \in V$ such that $u \fromto w v$ in $G_i$.
  (Observe that the two outgoing edges of $w$ must be in $T_{i-1},T_{i+1}$.)
  Let $G^{\fromto{}}$ be the subgraph of $G^2$ containing all edges occuring in
  $G_i^{\fromto{}}$ for each $i$.
\end{definition}
\begin{lemma}
  Let $G$ be a planar graph. For any graph $M$, if $M$ is a minor of
  $G_i^{\fromto{}}$ then $M$ is a minor of $\ssplit(G)$. 
\end{lemma}
\begin{proof}
  We will prove the following claim.
  \begin{claim}
  For any set $P = \{P_j\}$ of simple paths $P_j \subseteq V(G_i^{\gets\to})$,
  with endpoints $\{(s_j,t_j)\}$ such that no two paths $P_j,P_k$ have the same
  endpoints and $P_j \cap P_k \subseteq \{s_j,s_k,t_j,t_k\}$, there exists a set
  of paths $Q = \{Q_j\}$ of paths in $\ssplit(G)$ with the same endpoints such
  that
  \begin{align*}
  &Q_j \cap Q_k \subseteq \\&\;\{s_j,s_k,t_j,t_k\}
  \cup \{(s_j)_{i-1},(s_k)_{i-1},(t_j)_{i-1},(t_k)_{i-1}\}
  \cup \{(s_j)_{i+1},(s_k)_{i+1},(t_j)_{i+1},(t_k)_{i+1}\} \,,
  \end{align*}
  where the vertices $s_i,s_{i+1},s_{i-1}$ are defined as in the split
  operation.
  \end{claim}
  \begin{proof}[Proof of claim]
    For each path $P_j$, perform the following.  For each edge $\{u,w\}$ in the
    path $P_j$, there is some (not necessarily unique) vertex $v$ such that
    either $(v,u) \in T_{i-1}$ and $(v,w) \in T_{i+1}$, or the same holds with
    $u,w$ reversed. Add the edges $\{u,u_{i-1}\}, \{u_{i-1}, v_i\}, \{v_i,
    w_{i+1}\}, \{w_{i+1},w\}$ to $Q_j$. If $P_j$ is a singleton $P_j = \{u\}$ so
    $s_j = t_j$ then add $u$ to $Q_j$.

    Consider two paths $Q_j,Q_k$ constructed this way.  $G_i^\fromto{}$ has
    vertex set $V$ and $\ssplit(G)$ has vertex set $V' \supset V$.  By
    construction, $P_j \subseteq Q_j$ and $P_k \subseteq Q_k$ and $(Q_j \cap V)
    = P_j$. Suppose there exists $z \in Q_j \cap Q_k$ that is not an endpoint,
    so $z \notin \{s_j,s_k,t_j,t_k\}$. If $z \in V$ then $z \in P_j \cap P_k
    \subseteq \{s_j,s_k,t_j,t_k\}$, so we only need to worry about $z \in V'
    \setminus V$.
    
    If $z = v_i$ for some vertex $v$ then there are unique distinct vertices
    $u_{i-1},w_{i+1} \in V'$ adjacent to $v_i$ such that $u_{i-1},w_{i+1} \in
    Q_j \cap Q_k$. Then $u,w \in Q_j \cap Q_k$ also, so $u,w \in P_j \cap P_k$;
    but then $u \neq w$ are the start and end points of $P_j,P_k$, so $P_j =
    P_k$, a contradiction.

    If $z = v_{i-1}$ for some vertex $v \in V$ then $v \in Q_j \cap Q_k$, so by
    the case above, $v \in \{s_j,s_k,t_j,t_k\}$ and $z \in
    \{(s_j)_{i-1},(s_k)_{i-1},(t_j)_{i-1},(t_k)_{i-1}\}$.  Likewise for $z =
    v_{i+1}$.
  \end{proof}
  Let $M$ be a minor of $G_i^{\fromto{}}$, so a subdivision of $M$ occurs as a
  subgraph of $G_i^{\fromto{}}$. Therefore there is a set of paths $P$ in
  $G_i^{\gets\to}$ satisfying the conditions of the claim, so that by
  contracting each path into a single edge, and deleting the rest of the graph,
  we obtain $M$. Let $Q = \{Q_j\}$ be the set of paths given by the claim. For
  endpoints $s_j,t_j \in Q_j$, contract the edges $\{s_j,(s_j)_{i \pm 1}\}$ and
  $\{t_j,(t_j)_{i \pm 1}\}$. The result is a contraction of $\ssplit(G)$ and a
  set of paths $Q'$ that is a subdivision of $M$, so $M$ is a minor of
  $\ssplit(G)$, which proves the lemma.
\end{proof}
\begin{corollary}
  \label{cor:head to head is planar}
  $G_i^{\fromto{}}$ is planar and $G^{\fromto{}}$ has arboricity at most 9.
\end{corollary}
\begin{proof}
  A graph is planar iff it does not contain $K_5$ or $K_{3,3}$ as a minor
  (Kuratowski's Theorem). If $G_i^{\fromto{}}$ is not planar then it contains
  $K_5$ or $K_{3,3}$ as a minor, so by the above lemma, $\ssplit(G)$ contains
  $K_5$ or $K_{3,3}$ as a minor, so $\ssplit(G)$ is not planar, a contradiction.
  Since planar graphs have arboricity at most 3, the edge union $G^{\fromto{}}$
  of 3 planar graphs has arboricity at most 9.
\end{proof}
By separating the $\gets \to$ edges from the remaining edges of $\cP^2$, we
obtain a constant-cost universal SMP protocol for $\cP^2$, and then by applying
Theorem \ref{thm:universal smp to adjacency labeling} we obtain Theorem
\ref{thm:k distance for planar graphs}.
\begin{lemma}
  \label{lemma:planar graph protocol}
  For all $\epsilon > 0, R_\epsilon^\ob(\cP^2) =
  O\left(\log\frac{1}{\epsilon}\right)$.
\end{lemma}
\begin{proof}
  For a planar graph $G = (V,E)$ with a fixed planar map and a Schnyder wood $T_1, T_2,
  T_3$, define the graph $G_i = (V,E \setminus T_i)$ as the graph obtained by
  removing the edges in tree $T_i$.

  On planar graph $G \in \cP_n$ and vertices $x,y$, Alice and Bob perform the
  following:
  \begin{enumerate}
    \item For each $i$ define $x_i,y_i$ to be the parents of $x,y$ in
      $T_i$. Run the protocol for adjacency with error $\epsilon/7$ on
      $(x,y_i)$ and $(x_i,y)$ for each $i$.
    \item Run the protocol for low-arboricity graphs on $G^{\fromto{}}$ with
      error $\epsilon/7$.
    \item Accept iff one of the above sub-protocols accepts.
  \end{enumerate}
  By Corollary \ref{cor:head to head is planar}, $G^\fromto{}$ has arboricity at
  most 9, we may apply the protocol for low-arboricity graphs in step 2. If
  $\dist(x,y) > 2$ then the protocol will correctly reject with probability at
  least $1-\epsilon$ since there are 7 applications of $\epsilon/7$-error
  protocols.  It remains to show that if $\dist(x,y) = 2$ then the algorithm
  will accept.

  Suppose $x,y$ are of distance 2. Then the paths between them are of the
  following forms (with edge directions taken from the Schnyder wood).
  \begin{enumerate}
    \item $x \to v \to y$ or $x \to v \gets y$. This is covered by step 1.
    \item $x \gets v \to y$. This is covered by step 2. \qedhere
  \end{enumerate}
\end{proof}
Since planar graphs are an upwards family (just insert a new vertex), we obtain
a constant-size probabilistic universal graph for $\cP^2$.
\begin{corollary}
  For any $\epsilon > 0$, there is a graph $U$ of size $O(\log(1/\epsilon))$
  such that for every $G \in \cP^2$, $G \rembed_\epsilon U$.
\end{corollary}

\section{Discussion and Open Problems}

\subparagraph*{Error-tolerance.} In the introduction we mentioned that the universal
SMP model allows us to study error-tolerance in the SMP model. This could be
done as follows: suppose the referee knows a reference graph $G$ and the players
are guaranteed to see a graph that is ``close'' to $G$ by some metric. How much
does this change the complexity of the problem, compared to computing $G$? One
common distance metric in, say, the property testing literature, is to count the
number of edges that one must add or delete. That is, for two graphs $G,H$ on
vertex set $[n]$, write $\dist(G,H) = \frac{1}{n^2} \sum_{i,j \in [n]}
\ind{G(i,j) \neq H(i,j)}$.  The distance is usually thought of as a constant.
We can easily give a strong negative result for this situation:
\begin{proposition} Let $\cF$ be any family of graphs and $\cF_\delta$ the
family of graphs $G$ such that $\min_{F \in \cF} \dist(G,F) \leq \delta$. Then
\[
  R^\ob(\cF_\delta) = \Omega(\sqrt{\delta} n) \,.
\]
\end{proposition}
\begin{proof}
Let $G$ be any graph on $\sqrt \delta n$ vertices and let $F \in \cF$. Choose
any set $S \subseteq V(F)$ with $|S| = |G|$. Construct $F'$ by replacing the
subgraph induced by $S$ with the graph $G$. Then $\dist(F,F') \leq
\frac{|G|^2}{n^2} = \delta$ so $F' \in \cF_\delta$. Then the conclusion follows
from Theorem \ref{thm:lower bound for all graphs}.
\end{proof}
This suggests that this is not the correct way to model contextual uncertainty
in the SMP model, but universal SMP gives a framework for studying many other
error tolerance settings. For example, we could suppose that the referee knows a
reference planar graph $G$, and the players are guaranteed to see a graph $G'$
that is close to $G$ and also planar; this would not increase the cost of the
protocol due to our results on planar graphs.

\subparagraph*{Implicit graph conjecture.} A major open problem in graph labeling is
the \emph{implicit graph conjecture} of Kannan, Naor, and Rudich \cite{KNR92},
which asks if every \emph{hereditary} graph family $\cF$ (where for each $G \in
\cF$, every induced subgraph of $G$ is also in $\cF$) containing at most $2^{O(n
\log n)}$ graphs of size $n$ has an $O(\log n)$ adjacency labeling scheme. Not
much progress has been made on this conjecture (see e.g.~\cite{Spin03, Chan16}).
We ask a weakened version of this conjecture:
\begin{question}
  For every hereditary family $\cF = (\cF_n)$ such that $|\cF_n| \leq 2^{O(n
  \log n)}$, is $R^\ob(\cF) = O(\log n)$?
\end{question}
Good candidates for disproving the implicit graph conjecture are geometric
intersection graphs, like disk graphs (intersections of disks in $\bR^2$) or
$k$-dot product graphs (graphs whose vertices are vectors in $\bR^k$, with an
edge if the inner product is at least 1) \cite{Spin03}. These are good
candidates because encoding the coordinates of the vertices as integers will
fail \cite{KM12}.  Randomized communication techniques may be able to make
progress.

\subparagraph*{Modular lattices.} We have shown that there is no constant-cost
universal protocol for distance 2 in modular lattices but, like low-arboricity
graphs, adjacency (and therefore $O(\log n)$-implicit encodings) may still be
possible.

\subparagraph*{Planar graphs.} Our protocol for computing distance 2 on planar graphs
did not generalize in a straightforward fashion to distance 3. Nevertheless, we
expect that there is a method for computing $k$-distance on planar graphs with
complexity dependent only on $k$; given that a Schnyder wood partitions each
edge into 3 groups, we expect that $\widetilde O(3^k)$ should be possible, and
maybe only $\poly(k)$, considering that there is a $O(\sqrt n)$
distance-labeling scheme.

\subparagraph*{Sharing randomness with the referee.} Finally, it seems to be unknown
what the relationship is between SMP protocols where the referee shares the
randomness, and protocols where the referee is deterministic, even though both
models are used extensively in the literature.  Our Proposition \ref{prop:weak
to universal} relates these two models via universal SMP but does not yet give a
general upper bound on the universal cost in terms of the weakly-universal cost.
\begin{question} 
What general upper bounds can we get on universal SMP in terms of
weakly-universal SMP?
\end{question}

\begin{center}
  \textbf{\large Acknowledgments}
\end{center}
Thanks to Eric Blais for comments on the structure of this paper; Amit Levi for
helpful discussions and comments on the presentation; Anna Lubiw for an
introduction to planar graphs and graph labeling; Corwin Sinnamon for comments
on distributive lattices; and Sajin Sasy for observing the possible applications
to privacy. Thanks to the anonymous reviewers for their comments. This work was
supported in part by the David R.~Cheriton and GO-Bell Graduate Scholarships.

\bibliographystyle{alpha}
\bibliography{references}

\appendix

\newpage
\section{Appendix}
\begin{proof}[Proof of Proposition \ref{prop:embedding properties}]\leavevmode
\begin{enumerate}
  \item 
  If $A \embed B$ and $B \embed C$ with $\phi, \psi$ being the respective
  embeddings then for all $u,v \in V(A)$ we have
    $C(\psi \phi(u), \psi \phi(v)) = B(\phi(u), \phi(v)) = A(u,v)$.
\item 
  In the ``only if'' direction, it suffices to choose $G^\equiv$. 
  In the other direction, if $\phi : V(G) \to V(H)$ is an embedding and $\phi(u)
  = \phi(v)$ then for all $w \in V(G), G(u,w) = H(\phi(u),\phi(w)) =
  H(\phi(v),\phi(w)) = G(v,w)$ so $u \equiv v$.
\item 
  Let $g$ map a vertex of $G$ to its equivalence class and let $u,v \in V(G)$.
  If $G(u,v) = 1$ then $G^\equiv(g(u),g(v)) = 1$ by definition. If
  $G^\equiv(g(u),g(v)) = 1$ then there exists $u' \in g(u), v' \in g(v)$ such
  that $G(u',v') = 1$, so $G(u,v) = G(u',v) = G(u',v') = 1$.
\item 
  Let $g$ map vertices in $V(G)$ to their equivalence class and let $g(u),g(v)
  \in V(G^\equiv)$. If $g(u) \equiv g(v)$ then for any $w, G(u,w) = G^\equiv(g(u),g(w)) =
  G^\equiv(g(v),g(w)) = G(v,w)$ so $u \equiv v$ and therefore $g(u) = g(v)$.
  Therefore the map $g(u) \mapsto \{g(u)\}$ is an isomorphism $G^\equiv \to
  (G^\equiv)^\equiv$.
\item 
  If $G \embed H$ then by transitivity, $G^\equiv \embed G \embed H
  \embed H^\equiv$. Likewise, if $G^\equiv \embed H^\equiv$ then $G \embed
  G^\equiv \embed H^\equiv \embed H$. 
\item 
  If $G^\equiv$ is an induced subgraph
  of $H^\equiv$ then clearly there is an embedding. On the other hand, let $g(u),g(v) \in
  V(G^\equiv)$ be the
  equivalence classes of $u,v \in V(G)$ and suppose there
  is an embedding $\phi : G^\equiv \to H^\equiv$. If $\phi(g(u)) = \phi(g(v))$
  then $g(u) \equiv g(v)$ so $g(u) = g(v)$ since $(G^\equiv)^\equiv \simeq
  G^\equiv$. Therefore $G^\equiv$ is an induced subgraph of $H^\equiv$. \qedhere
\end{enumerate}
\end{proof}

\end{document}